\documentclass[journal]{IEEEtran}
\usepackage[dvips,final]{graphicx}
\usepackage{latexsym}
\usepackage{amssymb}
\usepackage{graphicx}
\usepackage{amsfonts}
\usepackage{amsmath}
\usepackage{latexsym}
\usepackage{multirow}

\usepackage{epsfig}
\usepackage{ifthen}
\usepackage{law}
\usepackage{graphics}
\usepackage{times}
\usepackage{cite}
\usepackage{algorithmic,color,algorithm}

\newtheorem{prop}{Proposition}
\newtheorem{defin}{Definition}

\begin{document}

\title{Distributed sensor failure detection in sensor networks}

\author{Tamara~To\v{s}i\'{c}, Nikolaos~Thomos\thanks{This work has been partly supported by the FNS project number PZ00P2-121906.} and Pascal~Frossard\\
{E}cole Polytechnique F\'{e}d\'{e}rale de Lausanne (EPFL)\\
Signal Processing Laboratory (LTS4), Lausanne, 1015-Switzerland\\
E-mail: \{tamara.tosic,nikolaos.thomos,pascal.frossard\}@epfl.ch.
} 


\date{\today}

\maketitle

\thispagestyle{empty}

\begin{abstract}
We investigate the problem of distributed sensors' failure detection in networks with a small number of defective sensors, whose measurements differ significantly from neighboring sensor measurements. Defective sensors are represented by non-zero values in binary sparse signals. We build on the sparse nature of the binary sensor failure signals and propose a new distributed detection algorithm based on Group Testing (GT). The distributed GT algorithm estimates the set of defective sensors from a small number of linearly independent binary messages exchanged by the sensors. The distributed GT algorithm uses a low complexity distance decoder that is robust to noisy messages. We first consider networks with only one defective sensor and determine the minimal number of linearly independent messages needed for detection of the defective sensor with high probability. We then extend our study to the detection of multiple defective sensors by modifying appropriately the message exchange protocol and the decoding procedure. We show through experimentation that, for small and medium sized networks, the number of messages required for successful detection is actually smaller than the minimal number computed in the analysis. Simulations demonstrate that the proposed method outperforms methods based on random walk measurements collection in terms of detection performance and convergence rate. Finally, the proposed method is resilient to network dynamics due to the effective gossip-based message dissemination protocol. 
\end{abstract}

\section{Introduction}\label{sec:intro}
Over the past years we have witnessed the emergence of simple and low cost sensors. This has led to wide deployment of sensor networks for monitoring 
signals in numerous applications, for example in medical applications or natural hazard detection. However, sensor networks have often a dynamic architecture with loose coordination due to the cost of communications. This raises new demands for collaborative data processing algorithms that are effective under network topology and communication constraints. In general, a sensor network is represented as a connected graph $\mathcal{G} = (\mathcal{V},\mathcal{E})$, where vertices $\mathcal{V}={\{s_i\}}_{i=1}^{S}$ stand for the $S$ sensors and edges $\mathcal{E}$ determine sensors' connectivity. For instance, if two sensors $s_i$ and $s_j$ lie within each other's communication range, the edge $e_{i,j} \in \mathcal{E}$ has a nonzero value.    
 Fig. \ref{fig:illus} illustrates a setup where sensors capture a smooth physical phenomenon (e.g., spatial temperature evolution) and generate messages that are eventually gathered for analysis. 
\begin{figure}[h!]
\begin{center}
\includegraphics[width=7cm]{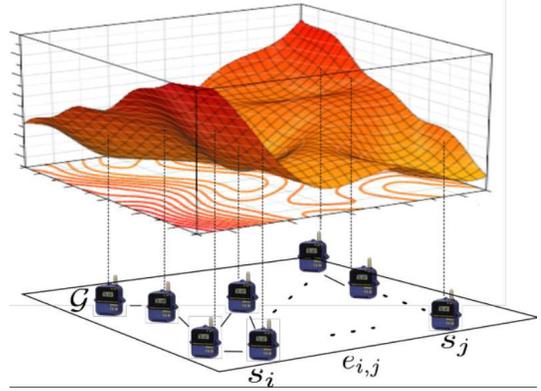}
\end{center}
\caption{\label{fig:illus}Ad-hoc sensor network measuring a smooth physical phenomenon.}
\end{figure}
When a sensor is defective, its measurements are inaccurate and can contaminate the signal analysis. It thus becomes important to detect the defective sensors in the network, so that their erroneous values do not impact the accuracy of the underlying data processing applications.

The detection literature can be mostly classified into centralized and distributed methods. Most of the works on detection methods for binary sparse signals mainly deal with centralized systems. The pioneering work in \cite{Dorfman:43} targets medical applications. It proposes a simple idea of pooling blood samples and observing the viral presence in a set, instead of performing tests on every single blood sample separately. Typically, the main target is to minimize the number of \textit{tests} required to identify all the infected samples, while keeping the detection procedure as simple as possible. This paradigm is known as Group Testing (GT) and has been proposed half a century ago. GT has been studied more recently in the context of sensor networks for detection of malicious events \cite{Young:11}. The detection approaches differ in scenarios with errors, inhibitors or combinations of them and the detection algorithms are rather naive \cite{Chen:08}. Defective sensors are detected  by iterative elimination of identified non-defective sensors from the test outcomes. The detection time is typically of order $\mathcal{O}(SB)$, where $B$ is the number of tests and $S$ is the total number of sensors. Particular test design methods improve the effective time for detection in centralized systems. For example, a useful test matrix property called $K$-disjunctness property (i.e., the Boolean sum of every $K$ columns does not result in any other column), speeds up the decoding process. This property is used in code designs, for e.g., for superimposed codes \cite{Dai:2009}, \cite{DeBonis:03}. Further, a random efficient detection is proposed in \cite{Indyk:10} with a decoding time of $\mathcal{O}(\pi(B)\cdot B \log^2 B + \mathcal{O} (B^2))$, where $B=\mathcal{O} (K^2 \log S)$ is the number of tests and $\pi$ denotes a polynomial. In our knowledge, this represents the state-of-the-art decoding performance.

In sensor networks, test design is contingent to the communication limitations. Works that consider constraints imposed by the sensor network topology in GT framework are not numerous. The authors in \cite{Cheraghchi:10} propose to form tests by a random walk process on well-connected graphs. The minimal number of tests required for detection in this case depends on the random walk mixing time. A bipartite graph structure is considered in \cite{Mezard:07} with a two-stage hybrid detection method. Here, a subset of defective items in the first stage is determined by pre-designed tests, while the remaining items are tested individually in the next step. Data retrieval for topology-adaptive GT is studied in \cite{Hong:04} where a binary tree splitting algorithm is proposed. The above methods use centralized decision algorithms which are not appropriate for large-scale sensor networks or networks with a dynamic topology because of the high communication costs. In those scenarios one rather needs to use distributed detection methods. To the best of our knowledge, however, no analysis on distributed detection methods which consider sparse and binary test signals are available. The distributed methods are rather employed for non-sparse signal detection with explicit network and message constraints. Such methods generally employ statistical decoders \cite{Varshney:97}.
For example, a Bayesian approach in \cite{Tsitsiklis:93} proposes to compute a detection score for a priori defined sets of hypothesis, which depends on the received messages. The hypothesis with the highest score drives the decision. The binary event detection problem for hierarchically clustered networks is proposed in \cite{Tian:07} where the cluster decisions are fused to make a final decision. Surveys on similar methods can be found in \cite{Viswanathan:97,Blum:97}. 

In this paper, we propose a novel distributed sensors' failure detection method that employs a simple distance decoder for sparse and binary signals. We assume that at most $K$ sensors are defective out of $S$ sensors in the network, where $K\ll S$. Therefore, the defective sensor identification problem boils down to a sparse binary signal recovery, where nonzero signal values correspond to defective sensors. Our approach is based on GT methods that are commonly applied for centralized systems. The core idea is to perform low-cost experiments in the network, called \textit{tests}, in order to detect the defective sensors. The tests are performed on pools of sensors by a set of sensors called \textit{master} sensors. The master sensors request sensor measurements from their neighbors. Each sensor responds to this request with probability $q$. Due to the smoothness of the measured function, non erroneous neighbor sensors typically have similar measurements. Each master sensor compares the sensor measurements based on a similarity measure (e.g., thresholding) to detect the presence of defective sensors in its vicinity. The result of this test takes a binary value, which might be possibly altered by noise. The tests and their outputs together form the network \textit{messages} that are communicated to neighborhood of the master nodes. The messages in the sensors are then disseminated in the network with a gossip algorithm (rumor mongering) \cite{Dimakis:10} that follows a pull protocol \cite{Demers:87}, \cite{Karp00}, \cite{Deb:2006}. Each time a new message reaches the sensor, its value is linearly combined with the message available at the current sensor in order to increase the diversity of information in the network. The message design and dissemination phases are repeated for several rounds. Due to the probabilistic test design and message dissemination we employ a simple distance decoder (e.g., Hamming decoder) that is able to detect defective sensors, as long as the number of messages is sufficient. We analyze the detection failure bounds and analytically derive the conditions needed for successful failure detection in the case of a single defective sensor. Then, we provide the error bounds for detection of multiple defective sensors. We show that the number of linearly independent messages required for detection is smaller in practice than the theoretical bounds obtained in our worst case analysis. We finally provide simulation results in regular and irregular networks. The experiments outline the advantages of the proposed detection method compared to other binary signal detection algorithms based on the random walk measurements gathering. Our algorithm outperforms random walk detection methods both in terms of the detection accuracy and convergence rate. 

This paper is organized as follows. Section \ref{sec:centr_det} reviews the centralized Group Testing framework. Section \ref{sec:decoder} proposes a novel distributed detection method. It describes the message formation and dissemination processes in sensor networks and discusses the detection problem for single and multiple defective sensors. Section \ref{sec:exper} presents the simulation results.

\section{Centralized detection with probabilistic Group Testing}\label{sec:centr_det}
We first review the centralized detection of sensor failures with methods based on GT. This framework is the ground for the novel distributed GT algorithm discussed in the next section. Detection is the identification of a subset of defective sensors whose measurements deviate significantly from those of the sensors in their vicinity. Based on the test construction, the methods for detection are categorized into \textit{deterministic} and \textit{probabilistic} algorithms. General centralized deterministic GT methods assign each sensor  to the set of tests prior to performing them, where the tests are designed to  assure detection. This approach however is not feasible for networks with large number of sensors. To alleviate this, probabilistic GT has been proposed in \cite{Cheraghchi:11}. We focus on test design methods that do not use the knowledge of realized test outcomes for novel test designs, since they are more appropriate in realistic settings.
 
Hereafter, we adopt the following notation: matrices and vectors are represented with boldface capital letters (\textbf{M}, \textbf{m}) and their elements are given with lowercase letters ($M_{i,j}, m_i$). Calligraphic letters are used to denote sets ($\mathcal{G}$), while $|\cdot|$ represents the number of elements in a set. The $i$-th column and the $i$-th row of $\mathbf{M}$ are represented with $\mathbf{M}_{:,i}$ and $\mathbf{M}_{i,:}$, respectively.
 
GT aims at detecting defective items in the set based on the outcome of binary tests. Nonzero entries of a $S$-dimensional binary vector $\mathbf{f}\in \mathbb{F}_2^S$ indicate the defective sensors. $\mathbb{F}_2$ is a finite field of size two and $\mathbf{f}$ is a $K$-sparse signal, where $K \ll S$. The tests preformed on sensor measurements are represented with a $B\times S$ dimensional matrix $\mathbf{W}$. The nonzero entries of $\mathbf{W}_{i,:}\in\mathbb{F}_2^S$ refer to the sensors that participate in the $i$-th test. The boolean matrix multiplication operator is denoted with $\otimes$. Then, the binary tests results are denoted with the test outcome vector $\mathbf{g}\in \mathbb{F}_2^B$:
\begin{equation}\label{eq:gt}
\bf{g} = \mathbf{W} \otimes \bf{f}.
\end{equation} 

The design of the matrix $\mathbf{W}$ is crucial for reducing the number of required tests for the detection of defective sensors. This design resembles the design of generator matrices of LDPC codes \cite{Gallager:62}. In the Tanner graph representation of LDPC codes, the LDPC encoded symbols are partitioned in check and variable nodes, where the check nodes are used to detect errors introduced during transmission of LDPC encoded symbols. Motivated by this similarity, the test matrix $\mathbf{W}$ is constructed as \cite{Cheraghchi:11}:
\begin{equation}\label{eq:weight_prob}
W_{i,j}=\left\{ \begin{array}{l l}
 1, & \mbox{with probability $q$,}\\
 0, & \mbox{otherwise.}\\ 
 \end{array} \right.
\end{equation}
The sensor participation probability is denoted with $q$. Such a design for the test matrix assures that with high probability, any test matrix column is not a subset of any union of up to $K$ columns (\textit{disjunctness} property). In other words, a matrix $\mathbf{W}$ is called $K$-disjunct if no column $\mathbf{W}_{:,i}$ of $\mathbf{W}$ lies in the sub-space formed by any set of $K$ columns $\mathbf{W}_{:,j}$ with $j \neq i$. This property enables fast decoding with a distance decoder (i.e., Hamming distance). 
The distance decoder exploits the knowledge of the test outcome vector $\mathbf{g}$ and the test matrix or the seed of the pseudorandom generator that has been used for generating the random test matrix. Next, we discuss in more details the disjunctnesss property and the detection probability in centralized GT, since they represent the starting point of the decentralized detection method proposed in the next section.

We first formally define the \textit{disjunctness} property \cite{Cheraghchi:11} of test matrices that results in low-cost detection. This property assures that the union of any set of at most $K$ different columns of $\mathbf{W}$ differs in at least $\epsilon$ positions from any other column of $\mathbf{W}$. 
\begin{defin}
\label{def:disjunct}
\textit{Disjunctness property}: A boolean matrix $\mathbf{W}$ with $S$ columns $\mathbf{W}_{:,1},\mathbf{W}_{:,2},\dots,\mathbf{W}_{:,S}$ is called $(K,\epsilon)$-disjunct if, for every subset $T$ of its columns, with $|T|\le K$:
\begin{equation}
\mid supp(\mathbf{W}_{:,i}) \backslash (\hspace{-2mm}\bigcup_{ j\in T\backslash\{i\} } \hspace{-2mm} supp(\mathbf{W}_{:,j}) ) \mid > \epsilon, \quad \forall i\in \{1,\dots,S\}
\end{equation}
where $supp(\mathbf{W}_{:,i})$ denotes the nonzero elements (support) of the column $\mathbf{W}_{:,i}$ and $\backslash$ is the set difference operator. 
\end{defin}

Disjunctness is an important property since it permits to analyze the detection probability. The connection between the structure of disjunct matrices and detection of defective items is given by the following proposition \cite{Cheraghchi:11}.
\begin{prop}
\label{def:noise_problem}
If the test matrix $\mathbf{W}$ fulfills a $(K,\epsilon)$-disjunct property, the detection problem is resolved in the $K$-sparse vector $\mathbf{f}$ with error parameter $\epsilon$. 
\end{prop}
The disjunct matrix parameter $\epsilon$ represents the distance decoder threshold for detection. The decoder accumulates the number of entries in a column of the $(K,\epsilon)$-disjunct test matrix that are different from the outcome vector $\mathbf{g}$. The columns of $\mathbf{W}$ that achieve the lowest Hamming distance correspond to defective sensors. For any column $\mathbf{W}_{:,i}$ of the test matrix $\mathbf{W}$ that is $(K,\epsilon)$-disjunct, the decoder verifies if:
\begin{equation}\label{eq:distdec}
\mid supp(\mathbf{W}_{:,i}) \backslash supp(\mathbf{g})\mid \le \epsilon,
\end{equation}
where $\mathbf{g}=\mathbf{W}\otimes\mathbf{f}$ is the vector of test outcomes. In other words, the decoder counts the number of positions in the column $\mathbf{W}_{:,i}$ for which the union of distinct columns differs from the set $T$ in order to detect defective items. The columns of the vector $\mathbf{f}$ are inferred as nonzero iff the inequality (\ref{eq:distdec}) holds.

Finally, the detection performance can also be analyzed in noisy settings, when the test matrix satisfies disjunctness property. The noisy settings results from the alternation of the nonzero entries in $\mathbf{W}$ with probability $1-p$, as represented in Fig. (\ref{fig:noise_bit_alternation}). 
\begin{figure}[h!]
\begin{center}
\includegraphics[width=0.35\columnwidth]{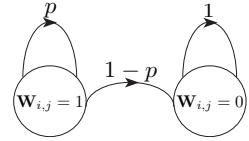}
\end{center}
\caption{\label{fig:noise_bit_alternation}Representation of noise influence of binary symbols in the test message. Non-zero values in the test matrix are flipped with probability $1-p$.}
\end{figure}

The following proposition provides the required number of measurements in centralized detection for successful decoding with the distance decoder in noisy settings \cite{Cheraghchi:11}.
\begin{prop}\label{th:nummeasur} 
Let the test matrix $\mathbf{W}$ be $(K,\epsilon)$-disjunct. The distance decoder successfully detects the correct support with overwhelming probability for a $K$-sparse vector $\mathbf{f}$ in a noisy environment when the number of tests is equal to 
$B=\mathcal{O}(K \log(S)/p^3).$
\end{prop}
The insights provided by the above results are used in the analysis of the novel distributed GT algorithm proposed in the next section.

\section{Distributed detection method}\label{sec:decoder}
\subsection{Sensor network message design and dissemination}\label{ssec:message_dissemin}
In this section, we propose a novel distributed failure detection algorithm and analyze its performance. The algorithm is based on a novel test design and message dissemination strategy in a distributed GT framework. The sensors iteratively create and disseminate messages in two-phases, denoted by $t_{I}$ and $t_{II}$. During the first phase $t_I$, the sensors obtain \textit{messages} that estimate the presence of defective sensors in their neighborhood. In the second phase $t_{II}$, the sensors linearly combine messages and exchange them employing a gossip mechanism. One round of our iterative distributed detection algorithm consists of these two phases. They are illustrated in Fig. \ref{fig:phases} and described below in more details.
\begin{figure*}[t!]
\begin{center}
\begin{tabular}{ccl}
~\includegraphics[width=0.4\columnwidth]{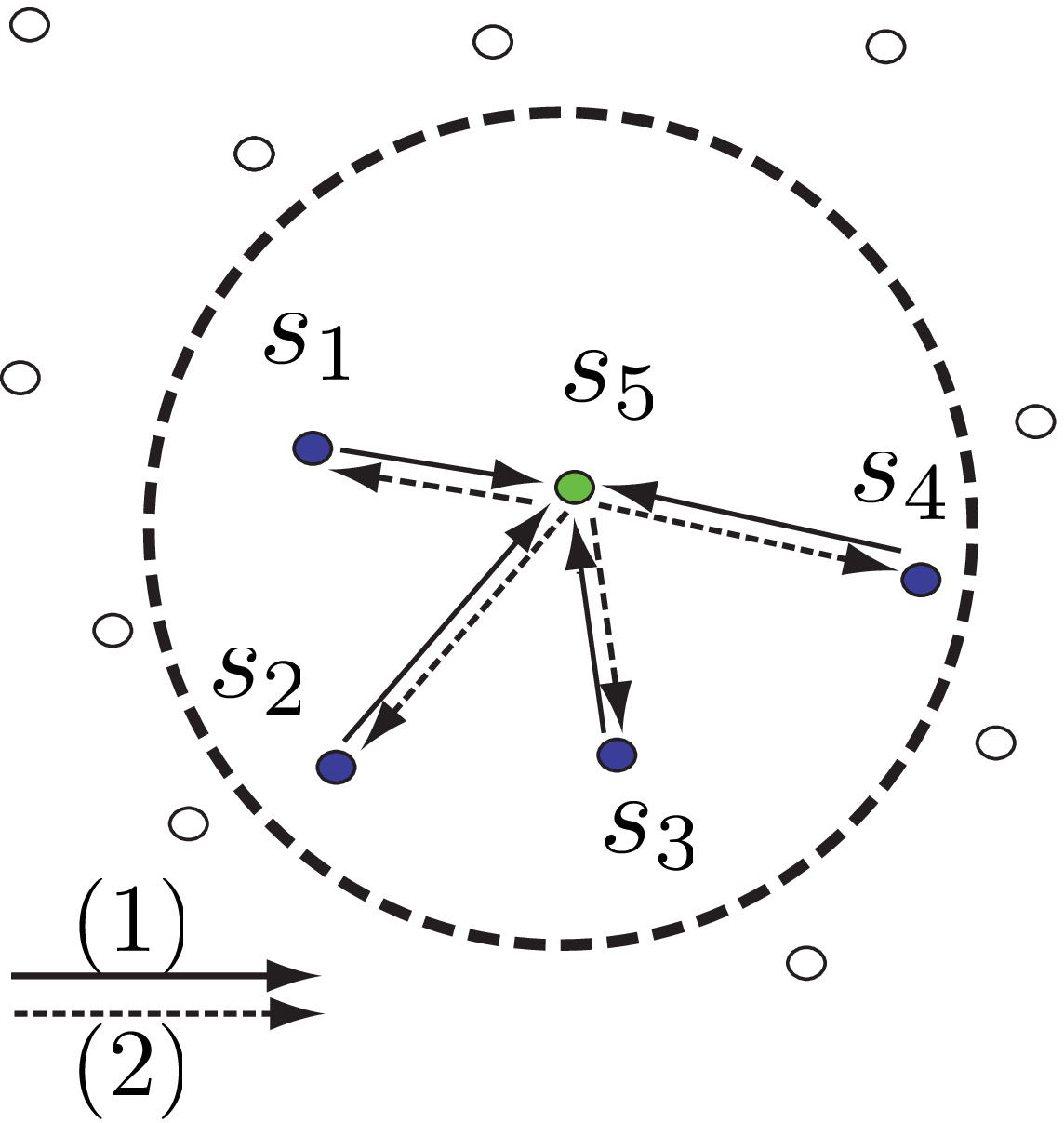}~&
~\includegraphics[width=0.4\columnwidth]{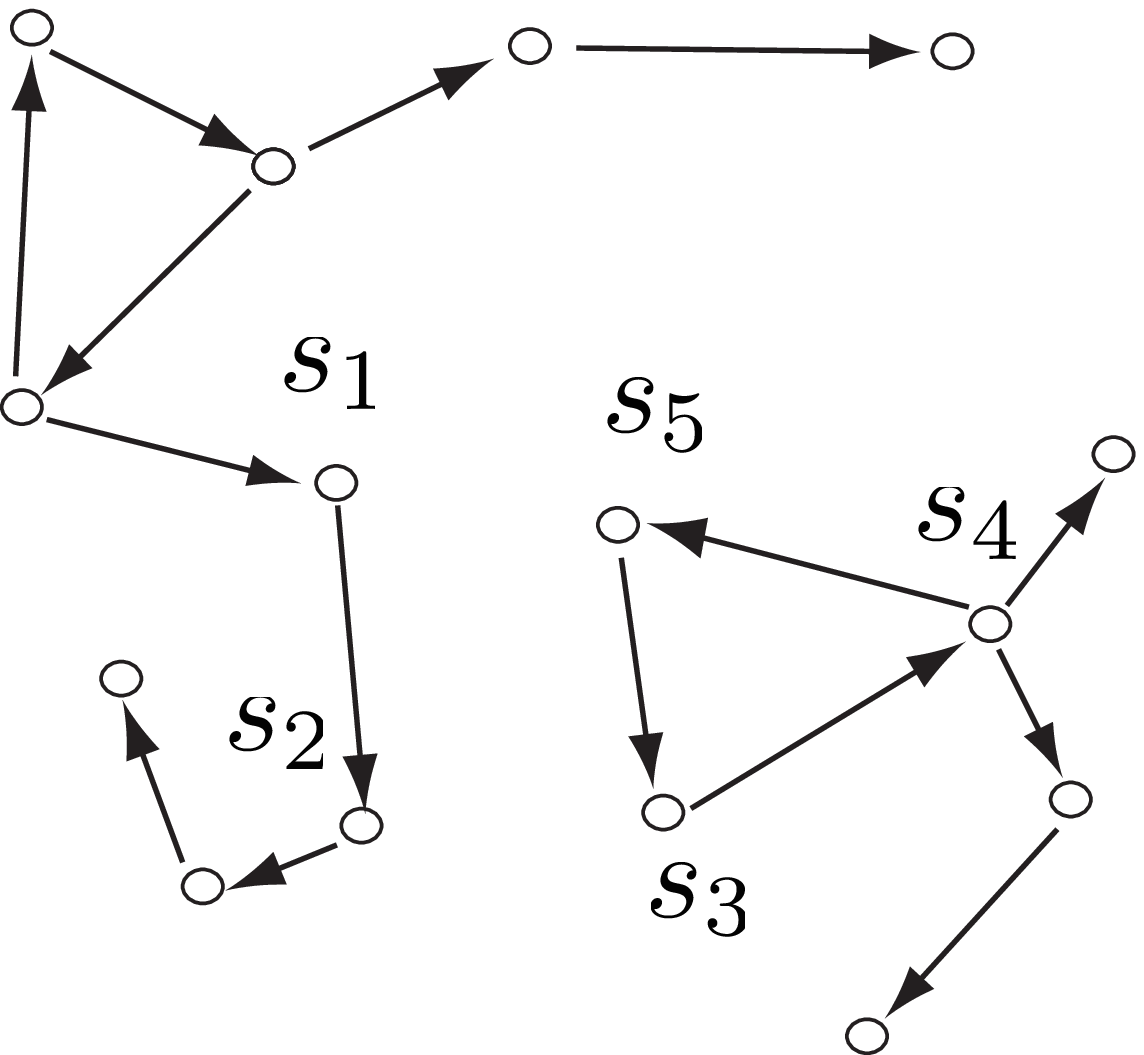}~&
~\includegraphics[width=0.5\columnwidth]{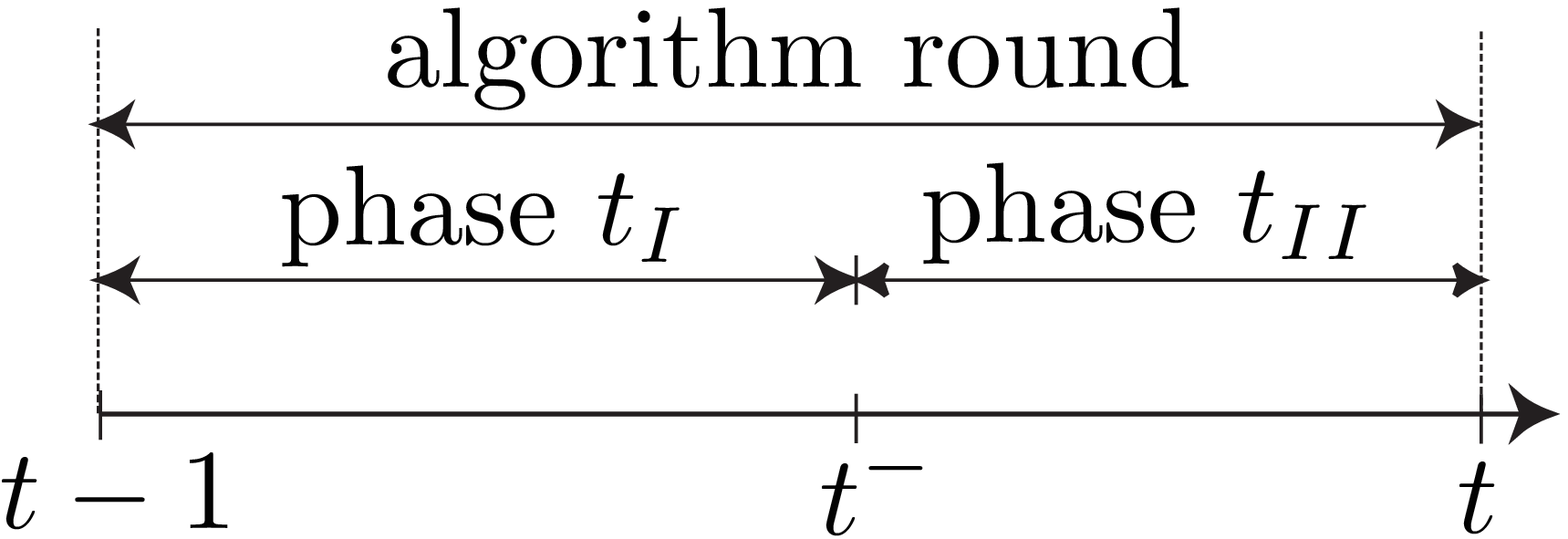}~\\
~(a) Phase $t_I$: Message design.~&~(b) Phase $t_{II}$: Message dissemination.~&~(c) Communication phases.~\\
\end{tabular}
\end{center}
\caption{
Illustration of the message design and dissemination through the sensor network. 
(a) Message formation based on local sensor measurements: Full and dashed arrows correspond to the steps of the message design, respectively. In the first step, the master sensor collects the sensor measurements from its neighbor sensors $\{{s_1},\dots,{s_4}\}$ and forms the message $(g_l(t^-),\mathbf{W}_{l,:}(t^-))$. In the second step, the message is propagated from the master sensor to its neighbor sensors.
(b) Message dissemination based on a gossip algorithm with pull protocol, where the sensors request the messages from their neighbors chosen uniformly at random. (c) Rounds of communication in our iterative detection algorithm consist of the message design ($t_I$) and the message dissemination ($t_{II}$) phases.
}
\label{fig:phases}
\end{figure*}

The first phase $t_I$ in round $t$ represents the message construction process illustrated in Fig. \ref{fig:phases}(a). $L$ master sensors cluster the network into disjoint subsets $\mathcal{V}_l\subset \mathcal{V}$, $l=1,\dots, L$. Clustering is used to bound the search space of decoder, as explained in the following subsections. 
Measurements of neighbor sensors do not vary significantly when the sensors are not defective when the signal under observation is smooth over the sensor field. The master sensors locally gather the readings or measurements from sensors that participate in their test.   Each sensor randomly participates in the test with probability $q$, as given in Eq. (\ref{eq:weight_prob}). The master sensor estimates the presence of defective sensors within its neighborhood and then attributes a binary value $f(\mathbf{s}_i)\in \mathbf{f}$ to each sensor in the neighborhood. The value $f(s_i) = 1$ denotes that the sensor $s_i$ is defective. Noise alternates non-zero bits with the probability $1-p$, as shown in Fig. (\ref{fig:noise_bit_alternation}). The test outcome at master node $l$ is finally computed as:
\begin{equation} \label{eq:onetest}
g_{l} =\mathbf{W}_{l,:}\otimes\mathbf{f}=\left\{
\begin{array}{l l}
 1, & \mbox{sensor(s) $\in\mathcal{K}$ ,}\\
 0, & \mbox{otherwise,}\\
 \end{array} \right.
\end{equation}
where the binary matrix operator $\otimes$ is composed by $\odot$ and $\oplus$ and stand respectively for the bitwise OR and the bitwise addition operators, where $\mathcal{K}$ is the set of defective sensors. The message formed by a master sensor $l$ during the phase $t_I$ consists of the outcome $g_{l}$ and the test participation identifier $\mathbf{W}_{l,:}$. The message $(g_{l}(t^-),\mathbf{W}_{l,:}(t^-))$ is sent to the neighbor sensors, which concludes the phase $t_{I}$. 

During the phase $t_{II}$, the messages created in the phase $t_I$ are disseminated within the network. The phase $t_{II}$ is illustrated in Fig. \ref{fig:phases}(b). Every sensor $i\in\{1,\dots,S\}$ requests the message formed at the previous round from its neighbor $j$, chosen uniformly at random, following a gossip mechanism with pull protocol. Next, each sensor $j$ responds to the message request that it has received from sensor $i$ by sending its message from the previous round. This process is performed only once per round. The sensor $i$ further combines these messages as follows:
\begin{eqnarray}\label{eq:diss}
g_i{(t)} \leftarrow g_{{i}}{(t^-)}\oplus g_{j}{(t-1)},\nonumber\\
\mathbf{W}_{i,:}(t) \leftarrow \mathbf{W}_{i,:}{(t^-)} \oplus \mathbf{W}_{j,:}{(t-1)},
\end{eqnarray}
where $g_j{(t-1)}$ denotes the sensor outcome value of the neighbor $j$ at the previous round $(t-1)$. The vector $\mathbf{W}_{i,:}(t)$ represents the test indicator vector at the sensor $i$ in round $t$. Since the messages are created probabilistically, the message combination in the different rounds assures that an innovative message reaches sensors at every round with high probability. A toy example of the dissemination phases is illustrated in Fig. \ref{fig:mess_form}. In this example the sensor $s_2$ at round $t$ pulls the message from the sensor $s_1$ and constructs a new message according to Eq. (\ref{eq:diss}). 
\begin{figure}[htb]
\begin{center}
\includegraphics[width=0.8\columnwidth]{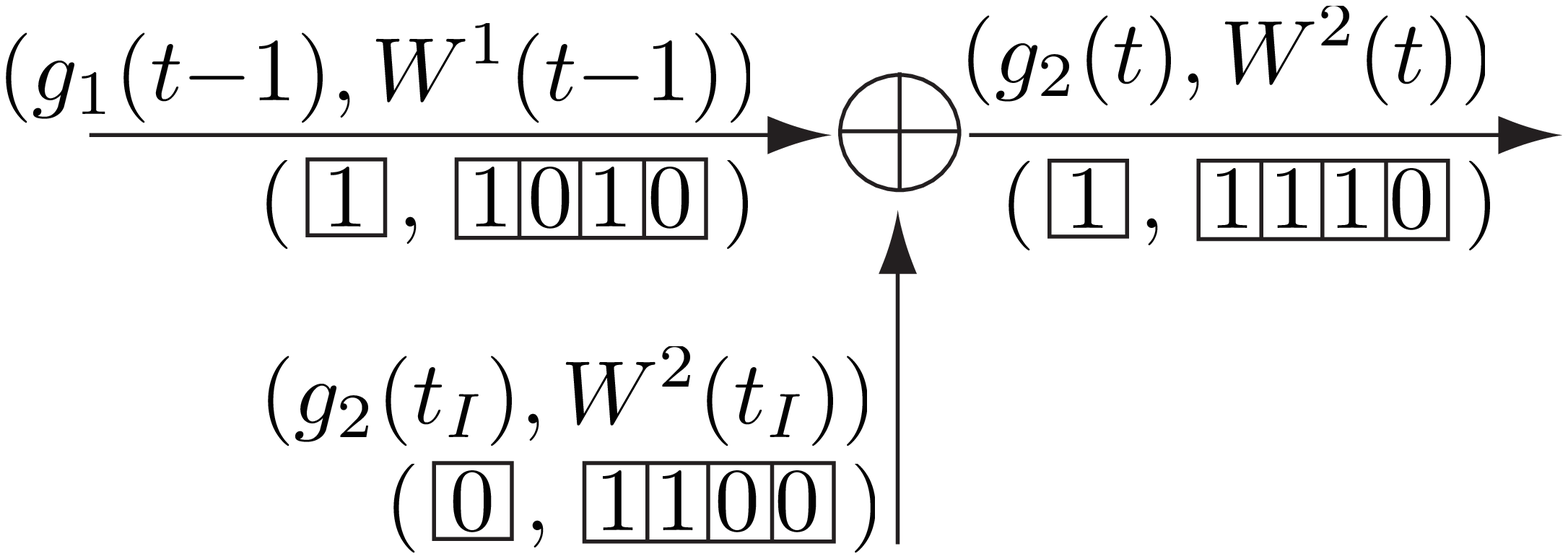}
\end{center}
\caption{
The message formation at sensor $s_2$ in round $t$. We assume that sensor $s_2$ pulls sensor $s_1$ to send its previous round values (round $t-1$). We assume that the sensor $s_3$ is defective $\mathbf{f}=[0 0 1 0 \dots]$. 
 The outcome value and the test identifier vector are formed by bitwise XOR.
 }
\label{fig:mess_form}
\end{figure}

In a matrix form, the process of message formation and transmission in $B$ rounds of our algorithm at any sensor in the network is represented as:
\begin{equation}\label{eq:prob}
\bf{g} =\bf{W} \otimes \bf{f},
\end{equation}
where the sensor identifier matrix $\mathbf{W}=[\mathbf{W}_{1,:}{(t)};\dots;\mathbf{W}_{B,:}{(t)}]$ is of size $B \times S$. 
The latter equation resembles to the outcome computation in the centralized GT case. However, in the distributed GT the tests represent linear combinations of test vectors that build disjunct matrix with high probability, as given in Eq. (\ref{eq:weight_prob}). To make a clear distinction between test matrices in proposed and centralized setup, we assume that an oracle has a direct access to the master nodes. Let $\mathbf{C}_{i,:}$ denote the concatenation vector of test realizations at master nodes collected by an oracle in the phase $t_I$ of the round $t=i$. The matrix $\mathbf{C}=[\mathbf{C}_{1,:}; \mathbf{C}_{2,:}\dots\mathbf{C}_{B,:}]$ then represents the test matrix over $B$ collection rounds. Observe that the matrix $\mathbf{C}$ is by construction disjunct, while $\mathbf{W}$ is built on the boolean addition of rows of $\mathbf{C}$ as in Eq. (\ref{eq:diss}). The values in $\mathbf{W}$ thus depend on the random message propagation path, which is obviously not the case in the centralized GT algorithm. Note that, for an arbitrary network, the number of network rounds required for collecting a particular number of linearly independent tests varies and depends on the network topology, the number of master nodes $L$ and the test participation probability $q$. 

Once every sensor has gathered enough test messages, it independently solves the failure detection problem finding the binary vector $\mathbf{f}$ that satisfies the tests in Eq. (\ref{eq:prob}). This solution $\mathbf{f}$ indicates the defective sensors. This process is analyzed in more details below.

\subsection{Detection of one defective sensor in the network}\label{sec:caseI_sol}
We first analyze the case of a single defective sensor (case $K=1$) in the network and study the detection probability of our distributed algorithm. To recall, the distance decoder used for detection computes the Hamming distance between two vectors $\mathbf{a}$ and $\mathbf{b}$. The element-wise distance is given by:
\begin{equation}
dist(a_i,b_i)=\left\{ 
\begin{array}{l l}
 1, & \mbox{if $a_i \ne b_i$,}\\
 0, & \mbox{otherwise.}\\ 
 \end{array} \right.
\end{equation}
To avoid the false alarms, the decoder threshold $\epsilon$ is set to the value that is higher than the expected number of noise-induced bit flips per columns in the disjunct matrix $\mathbf{C}$  \cite{Cheraghchi:11}:
\begin{equation}\label{eq:epsilon}
\epsilon=(1+\delta)(1-p)qB.
\end{equation}
where $\delta>0$ is a small constant and $B$ is the number of rows in $\mathbf{C}$. Columns of $\mathbf{C}$ have in average $qB$ non-zero elements. Every non-zero matrix element is flipped with probability $(1-p)$ and the expected number of flips per column is:
\begin{equation}\label{eq:mu}
\mu=(1-p)qB. 
\end{equation}
Recall that the matrix $\mathbf{C}$ is by construction a disjunct matrix. Proposition \ref{def:noise_problem} states that the detection problem is resolved for tests that form a disjunct test matrix. However, the messages available at sensors in the network form a test matrix that is obtained by linear combinations of disjunct matrix rows and not disjunct matrix rows itself. Nevertheless, we show below that the distance decoder detects defective sensor with high probability under certain conditions. 

The formal propositions for detection with high probability are given below. First we show that the proposed algorithm in the network with a single master node designs a $(K,\epsilon)$-disjunct matrix $\mathbf{C}$ during the phase $t_I$. Next we show that in a single cluster network linear combinations of rows in $\mathbf{C}$ preserve distances between the test outcome and the column of the defective sensor in the test matrix. We then build on these two propositions to analyze the number of messages needed for the  distributed detection of a single defective sensor, which is given in Proposition \ref{th:mytheoremII}. 

We first show that for a network with a single master node ($L=1$) and probabilistic message design in the phase $t_I$, a $(K,\epsilon)$-disjunct matrix $\mathbf{C}$ is built with high probability. This case boils down to the centralized collection of data described in \cite{Cheraghchi:11} and the defective sensor can be detected by a distance decoder as shown in Proposition \ref{def:noise_problem}. 
\begin{prop} For a single-cluster network, the message design over the phase $t_I$ of our proposed method builds a $(K,\epsilon)$-disjunct matrix $\mathbf{C}$ with high probability for an arbitrary $K$ and $\epsilon$ defined as in Eq. (\ref{eq:epsilon}).
\label{prop:one}
\end{prop}
\begin{proof}
We show that the probability that the number of rows with a good disjunctness property $G$ of $\mathbf{C}$ is smaller than $\epsilon$ and we follow the development proposed in \cite{Cheraghchi:11}. The sensor participation probability $q$ in a test defined as in Eq. (\ref{eq:weight_prob}). A row of the matrix $\mathbf{C}_{i,:}$ is considered to have a good disjunctness property if a single symbol ``$1$'' occurs, while the rest $K-1$ values are equal to zero. The probability of such an event is equal to $\mu_1=q(1-q)^{K-1}$. The random variable that marks the total number of rows with such a property is denoted with $G$. The distribution of $G$ is binomial with a mean value $\mu_2=\mu_1B$. We show that the probability of having less than $\epsilon$ rows with good disjunctness property is small under the assumption that $\epsilon< \mu_2$. We limit this probability by a Chernoff bound as:
\begin{equation}
P(G < \epsilon) \le e^{-\frac{1}{2}\frac{(\mu_2-\epsilon)^2}{\mu_2}}= e^{-q B \frac{[(1-q)^{K-1}-(1-p)(1+\delta)]^2}{2(1-q)^{K-1}}}.
\end{equation}
Knowing that $2< e < 3$ and that constant $\alpha\ge 0$, we get $2^{-\alpha} \ge e^{-\alpha} \ge 3^{-\alpha}$. Since $3^{-\alpha} \le (1+\frac{-\alpha}{K})^K \le 2^{-\alpha}$ holds, $\gamma=\frac{[(1-q)^{K-1}-(1-p)(1+\delta)]^2}{2(1-q)^{K-1}}$ is bounded. For the parameter choice in \cite{Cheraghchi:11} $(\delta, \alpha)=(\frac{p}{2},\frac{p}{8})$, the value $\gamma=\mathcal{O}(p^3)$. Therefore this probability can be designed to be arbitrary small: 
\begin{equation}\label{eq:cond2}
P(G < \epsilon) \le e^{-B \gamma/K}=e^{-\mathcal{O}( B p^3/K)}.
\end{equation}
\end{proof}
Then we show that linear combinations of rows of $(K,\epsilon)$-disjunct matrices $\mathbf{C}$ in a network with a single master node preserve the Hamming distance only between the column of matrix $\mathbf{W}_{:,k}$ that corresponds to the defective sensor $s_k$ and the outcome vector $\mathbf{g}$.
\begin{prop}
Let $\mathbf{C}$ be the $(K,\epsilon)$-disjunct matrix created over consecutive $B$ rounds in a single-cluster network during the phase $t_I$. Linear combinations of messages generated during the phase $t_{II}$, performed as in Eq. (\ref{eq:diss}), preserve the Hamming distance between the column of obtained matrix $\mathbf{W}_{:,k}$ that corresponds to the defective sensor $s_k$ and the outcome vector $\mathbf{g}$. 
\label{prop:two}
\end{prop}
\begin{proof}
We first analyze the case that leads to a decoding failure for $(K,\epsilon)$-disjunct matrices following a development similar to \cite{Cheraghchi:11}. We prove further that linear combinations of rows in such matrices preserve vector distances between the outcome vector and the column of $\mathbf{W}$ that corresponds to the defective sensor. 

A decoding failure with a distance decoder occurs in a $(K,\epsilon)$-disjunct matrix when the number of flips of column elements of $\mathbf{C}$ is higher than $\epsilon$. The probability of  occurrence of a single flip is equal to $\mu_3=q(1-p)$. Let $F$ denotes the number of flips in the columns of the matrix. Hence, the expected number of flips per column is given in Eq. (\ref{eq:mu}). We want to compute the lower bounds for the event that more than $(1+\delta) \mu$ flips occurred in the column of the matrix, where $\delta>0$. Applying the Markov inequality:
\begin{equation}
P(F\ge (1+\delta)\mu ) \le \inf_{d>0} \frac{\prod_{i=1}^{S}E[e^{dF_i}]}{e^{d (1+\delta)\mu}}
\end{equation}
and plugging the probability of the single flip event:
\begin{equation}
P(F_i)=\left\{ 
\begin{array}{l l}
 1, & \mbox{with probability $(1-p)q$,}\\
 0, & \mbox{with probability $1-(1-p)q$,}\\ 
 \end{array} \right.
\end{equation}
to the expectation term of the previous equation leads to: 
\begin{eqnarray}
P(F\ge (1+\delta)\mu ) & \le & \inf_{t>0} \frac{\prod_{i=1}^{m_i} [(1-p)q e^{d}+(1-(1-p)q)]}{e^{d (1+\delta)\mu}} \nonumber \\ & = & \inf_{t>0} \frac{\prod_{i=1}^{m_i} [(1-p)q (e^{d}-1) +1]}{e^{d (1+\delta)\mu}}.
\end{eqnarray}
If we set $(1-p)q (e^{d}-1) = x$ and plug the inequality $1+x< e^x$, we obtain:
\begin{eqnarray}
P(F\ge (1+\delta)\mu ) & \le & \inf_{d>0} \frac{\prod_{i=1}^{m_i} e^{(1-p)q (e^{d}-1)}}{e^{d (1+\delta)\mu}} \nonumber \\ & = &\inf_{t>0} \frac{e^{(1-p)q m_i (e^{d}-1)}}{e^{d (1+\delta)\mu}} \nonumber \\ & = & \inf_{t>0} \frac{e^{\mu(e^{d}-1)}}{e^{d(1+\delta)\mu}}.
\end{eqnarray}
For the constant $d=log(1+\delta)$, we finally obtain:
\begin{equation} \label{eq:prev_app}
P(F\ge (1+\delta)\mu ) \le ( \frac{ e^{ ^\delta } }{ (1+\delta)^{(1+\delta)} })^{\mu} =e^{ \mu\delta - \mu(1+\delta) \log(1+\delta) }.
\end{equation}
Observing that $log(1+\delta) > \frac{2\delta}{2+\delta}$, the Eq. (\ref{eq:prev_app}) becomes:
\begin{equation} \label{eq:cond1} 
P(F\ge (1+\delta)\mu ) \le e^{\frac{-\mu \delta^2}{2+\delta}}.
\end{equation}
 
The outcome value $\mathbf{g}$ depends on the presence of a defective sensor $s_k$ in the test. We prove here that the distance between $\bf{g}$ and the $k$-th column $\mathbf{W}_{:,k}$ does not increase more than $\epsilon$ during $t_{II}$, while this is not true for the rest of the columns. When sensor $j$ sends its message to sensor $i$ during the round $t$, we have:
{\small
\begin{eqnarray}\label{eq:distances}
&&\hspace{-10mm}dist\Big(g_i{(t)},W_{i,k}{(t)}\Big) \nonumber \\ 
&&\hspace{-8mm}=dist\Big(g_i{(t^-)}\oplus g_j{(t-1)},W_{i,k}{(t^-)} \oplus W_{j,k}{(t-1)}\Big) \nonumber\\ 
&&\hspace{-8mm}= dist\Big(g_i{(t^-)},W_{i,k}{(t^-)}\Big) \oplus dist\Big(g_j{(t-1)},W_{j,k}{(t-1)}\Big),
\end{eqnarray} }
where the first equality results from Eq. (\ref{eq:diss}). The second equality follows directly from the fact that the values of $\mathbf{g}(t^-)$ and the columns $\mathbf{W}_{:,k}(t^-)$ are identical for the defective sensor due to Eq. (\ref{eq:prob}). Since these two columns can initially differ at $\epsilon$ positions due to noise flips, the overall distance between the vectors $\mathbf{g}(t^-)$ and $\mathbf{W}_{:,k}(t^-)$ is at maximum $\epsilon$ given in Eq.(\ref{eq:epsilon}).
\end{proof}
We consider now networks with $L$ master sensors and an hypothetical centralized data collection. We assume that $L$ master nodes cluster the sensor network in disjoint subsets, where every sensor belongs to exactly one cluster. The master nodes perform message design over the rounds $t_I$ as proposed by our algorithm. We show now that the tests gathered from the $L$ different clusters build a disjunct matrix, where each cluster relates a $(K,\epsilon)$-disjunct matrix. 
\begin{prop}
The diagonal matrix $\mathbf{C}=diag(\mathbf{C}_1,\dots,\mathbf{C}_L)$ obtained from $(K,\epsilon)$-disjunct matrices $\mathcal{C}=\{\mathbf{C}_i\}_{i=1}^L$  is at least $(K,\epsilon)$-disjunct.
\label{prop:three}
\end{prop}
\begin{proof}
Proof follows directly from the Definition \ref{def:disjunct} and the disjunctness property of the matrices in $\mathcal{C}$. 
\end{proof}
We consider now the gathering of messages that are linearly combined over successive rounds of our detection algorithm. Uniform gathering of linearly combined messages at $L$ clusters by a hypothetical centralized decoder results in detection of the defective sensor with high probability when the number of received messages is sufficient. 
\begin{prop}
When the $(K,\epsilon_i)$-disjunct matrices $\mathcal{C}=\{\mathbf{C}_i\}_{i=1}^L$ are linearly combined as in Eq. (\ref{eq:diss}), where $\epsilon=\sum_{i=1}^{L}\epsilon_i $ and $q=\sum_{i=1}^{L}q_i$, the resulting test matrix permits detection by a distance decoder with high probability as long as it contains in total $B \ge \mathcal{O}({K \log (S) /p^3})$ messages collected from clusters chosen at random.
\label{prop:four}
\end{prop}
\begin{proof}
We first show that a diagonal matrix constructed from $(K,\epsilon_i)$-disjunct matrices of the set $\mathcal{C}$ is $(K,\epsilon)$-disjunct. Next, we recall the Proposition \ref{prop:two} and finally, we show that the $B$ measurements assure a good disjunct property of cluster matrices. Let the number of rows for all matrices be $B=\mathcal{O}(K\log(S) /p^3)$. The parameters $\epsilon$  and $\epsilon_i$ are defined in Eq. (\ref{eq:epsilon}) and $\epsilon=\sum_{i=1}^{L}\epsilon_i=(1+\delta)(1-p)B q$, the diagonal matrix of $(K,\epsilon_i)$ matrices is $(K,\epsilon)$ disjunct. 
The next part of the proof follows from the Proposition \ref{prop:two} which states that a matrix whose rows are formed by linear combinations of rows of $(K,\epsilon)$-disjunct matrix permits detection with a distance decoder. Finally, we need to prove that for a given $\mathbf{C}_i$ the disjunct property holds given that at least $B$ messages are available. For this purpose, we follow a development similar to \cite{Cheraghchi:11} and consider maximum number of sensors in clusters is $S_{max}=S$. The probability bound given in Proposition \ref{prop:one} should hold for all possible choices of a fixed set of $T$ out of $S$ columns: $\cup_{T} P(G\le\epsilon) \le S e^{-B q \gamma}$. This probability can be arbitrary small, e.g., in case $B \ge \frac{K \log S}{\alpha \gamma}= \mathcal{O}(K \log S/p^3)$. Further on, the condition in Eq. (\ref{eq:cond1}), which gives the probability bound that the number of flips in any $K$ out of $T$ columns exceeds a threshold value $\epsilon$ is also bounded. It reads $$\cup_{K}P(F\ge (1+\delta)\mu)\le K e^{\frac{- \delta^2}{2+\delta} \mu} = K e^{\frac{- \delta^2}{(2+\delta)p^3} (1-p)q K log(S) },$$ where the last equality is obtained by using Eq. (\ref{eq:mu}). This probability is small for the sufficiently large value of $B=\mathcal{O}(Klog(S)/p^3)$. 
\end{proof}

We now analyze the proposed distributed algorithm and consider the detection requirements for every sensor in the network. We show that the test messages collected by the sensors during the transmission rounds enable failure detection by the distance decoder with high probability if the number of messages is sufficient, where the decoder operations are performed locally at sensors.
\begin{prop}
We assume that $L$ master sensors partition the sensor network in disjunct parts. Test realizations within a cluster form test vectors. Over the rounds, these vectors create $(K,\epsilon)$-disjunct matrices $\mathcal{C}=\{\mathbf{C}_i\}_{i=1}^{L}$: 
\vspace{-1mm}
\begin{equation}
\mathbf{C}_i=\left\{
\begin{array}{l l}
 1, & \mbox{with probability $q_i={\alpha_i}$,}\\
 0, & \mbox{otherwise,}\\
 \end{array} \right.
\end{equation}
where $q=\sum_{i=1}^{L}q_i$. Messages $(\mathbf{g}_i,\mathbf{W}_{i,:})$ arrive at all the sensors in the network in our proposed algorithm, as described in the previous section. If the above assumptions hold and if the number of linearly independent messages received per cluster at every sensor in the network is at least $B/L$, where $B\hspace{-1mm}\ge\hspace{-1mm} \mathcal{O}(K \log(S)/p^3)$, the probability that sensors fail to detect the defective sensor by the distance decoder tends to zero as $S\rightarrow \infty$.
\label{th:mytheoremII}
\end{prop}
\begin{proof}
The message collection method does not influence the decoder performance, since the number of per-cluster measurements is sufficient for decoding with high probability. Therefore, the proof follows from the proof of Proposition \ref{prop:four}.
\end{proof}

\subsection{Detection of multiple defective sensors in the network}\label{sec:multipledef}
We analyze now the distributed detection of multiple defective sensors, where the number of defective sensors is much smaller than the total number of sensors. We propose here to slightly modify our distributed algorithm and to limit the decoder search space to be able to apply the Hamming distance decoder. The protocol modification and the adaptation of the distance decoder are described below. We assume that sensors completely differentiate between sensors in the network that belong to particular clusters and that at most one defective sensor is located in a given cluster. This knowledge limits the size of the decoder search space. 

The proposed protocol is first modified as follows to deal with multiple defective sensors. A decoder error occurs when two or more messages with positive test outcomes are combined together during the phase $t_{II}$, since the distance preserving property defined in Eq. (\ref{eq:distances}) is not guaranteed in this case. Since the number of defective sensors is very small compared to the total number of sensors, this event however occurs rarely. We explain the protocol modification with a simple example. Let the sensor $i$ pull the message from the sensor $j$, where both sensor test outcomes have nonzero values. Instead of combining the messages as in Eq. (\ref{eq:diss}), we simply buffer the new message of sensor $i$ and consider the message from sensor $j$ at previous round as the final outcome of the phase $t$:
\begin{eqnarray}\label{eq:diss_modified}
g_i{(t)} =  g_{j}{(t-1)},\nonumber\\
\mathbf{W}_{i,:}(t)= \mathbf{W}_{j,:}{(t-1)}.
\end{eqnarray}
At the first subsequent round $\tau\ge t+1$ of our distributed algorithm where both messages $g_i{(\tau)}$ and $g_{j}{(\tau-1)}$ have non-zero values as test outcomes, $g_i{(\tau)}$ is replaced by the message buffered in node $i$. The rest of the protocol remains unchanged. 

Then the decoding proceeds in two main steps. First, the appropriate unions of test matrix columns are created to form a search set space and second, the Hamming distance between the test outcome vector and the vectors of the search set are computed. The minimum Hamming distance indicate the solution of the detection problem. The outcomes $\mathbf{g}=[ \mathbf{g}_0 \,\mathbf{g}_1]^T$ collected at some sensor are divided into two sets, i.e., the negative and positive outcome vectors $\mathbf{g}_0$ and $\mathbf{g}_1$, respectively. Subsequently, the rows of the test matrix $\mathbf{W}$ form two sub-matrices $\mathbf{W}_0$ and $\mathbf{W}_1$ and Eq. (\ref{eq:prob}) is rewritten as:
\begin{equation}
 \left[ 
\begin{array}{c}
\mathbf{g}_0\\
\mathbf{g}_1 \\
\end{array} 
\right] = 
 \left[ 
\begin{array}{cc}
\mathbf{W}_0 & 0\\
0& \mathbf{W}_1\\
\end{array} 
\right]
 \left[ 
\begin{array}{c}
\mathbf{f}_0\\
\mathbf{f}_1 \\
\end{array} 
\right].
\end{equation}
We eliminate non-defective sensors from $\mathbf{W}_1$ using the knowledge from $\mathbf{W}_0$ and obtain $\mathbf{W}_1^{'}$. The columns of interest are those columns of $\mathbf{W}_1^{'}$ which contain at least one non-zero value. These columns are classified in sets $\mathcal{H}$, whose size depends on the complete or partial sensor knowledge about cluster affiliation of other sensors in the network. Columns belonging to the same cluster are grouped together in a set $\mathcal{H}_i$, where $i\in\{1,\dots,L\}$ and $L$ is the number of clusters. The search space $\mathcal{U}$ consists of vectors that are obtained from unions of up to $K$ columns, where each column is picked from a different set $\mathcal{H}_i$. We choose up to $K$ columns, since the number of defective elements can be smaller than $K$ by the problem definition, while the selection of at most one column from a particular $\mathcal{H}_i$ comes from the assumption that at most one defective sensor exists in each cluster. For instance, let the number of defective sensors and clusters be $(K,L)=(2,2)$. Let $\mathcal{H}_1$ contain $h_1$ and $\mathcal{H}_2$ contain $h_2$ columns. Then the search space size has in total $h_1h_2+h_1+h_2$ elements, where $h_1 h_2 = \binom{h_1}{1} \cdot {\binom{h_2}{1}}$ denotes the number of unions of $K=2$ columns and single column subsets are chosen in $h_1+h_2$ ways. Distance decoding is performed between $\mathbf{g}_1$ and elements of the set $\mathcal{U}$, starting from the vectors that are created as unions of $K$ columns towards the smaller number of column unions. If no solution exists for a particular value of $K$, we perform the decoding for vectors built from $K-1$ column unions of $\mathcal{H}_i$. If no unique solution is found, we encounter a decoding failure. 

Now that the decoder has been described, we analyze in details the number of required messages that are necessary for detection of multiple defective sensors with high probability.
\begin{prop}\label{prop:eight}
Under the assumption that at most one defective sensor is present in the cluster, that the number of available linearly independent messages at all sensors is at least $B/L$ per cluster, where $B\hspace{-1mm}\ge\hspace{-1mm} \mathcal{O}(K \log(S)/p^3)$ and that sensors know membership identifiers of all the clusters in the network, the distance decoder detects defective sensors at all sensors in the network with high probability.
\end{prop}
\begin{proof}
To recall, the transmission protocol ensures that the assumptions imposed by Proposition \ref{th:mytheoremII} hold for one defective sensor. Then, due to the assumption that at most one defective sensor is present in one cluster and that there is at most one defective sensor active in the test, we can form the set of solutions for the multiple defective case, which has a unique solution. Distance decoder between the outcome vector and a limited set of vectors that form a full search space can therefore find the appropriate solution. In other words, this procedure is identical to per-cluster decoding, where each cluster has at most one defective element, so the Proposition \ref{th:mytheoremII} can be applied. 
\end{proof}
\begin{prop}\label{prop:nine}
Under the assumption that one defective sensor at most is present in the cluster, that the number of available linearly independent messages at all sensors in the network is at least $B/L$ per cluster, where $B\hspace{-1mm}\ge\hspace{-1mm} \mathcal{O}(K \log(S)/p^3)$ and sensors know the partial set of identifiers of  the clusters in the network, the distance decoder detects defective sensors at all sensors in the network with high probability. 
\end{prop}
\begin{proof}
The search space $\mathcal{U}$ created in this case is larger but it contains the solution. Now the proof is identical to that in the previous proposition.
\end{proof}

Finally, we show that the assumption of at most one defective sensor occurrence per cluster is reasonable. We here bound the probability that at least two defective sensors occur within any cluster. An erroneous message is generated in a cluster that contains more than one defective sensor when only a fraction of defective sensors participate in the test actively and we denote the probability of such an event with $P(E)$. If defective sensors participate in the test, the distance within the column that signifies these vectors and the outcome result does not change. The same occurs if none of the defective sensors participate in a test. Due to the protocol modification, only one cluster may generate the erroneous message per round. In total we assume there are $m\in \{2,\dots,K\}$, $m\le n$ defective sensors and that clusters contain $n=\frac{S}{L}$ sensors. 
Then, the probability of decoding error in one cluster $P_{cl}(E)$ is equal to:
\begin{equation}
P_{cl}(E) = \sum_{m=2}^K P(n,q|m) P(m)=\sum_{m=2}^KP(n|m) P(q|m) P(m),
\end{equation}
due to independence of parameters $n$ and $q$. $P(m)$ represents the probability that some cluster contains $m$ defective sensors, $P(n|m)=\binom{n}{m}$ is a probability of choosing $m$ defective sensors within a cluster with $n$ sensors and $P(q|m)$ denotes the conditional probability of the error occurrence in a cluster with $m$ defective sensors and test participation probability $q$. We assume that $m$ takes a value from the set $\{2,\dots,K\}$ with uniform distribution, so $P(m)=\frac{1}{K-1}$. Next, $P(q|m)=1-q^{m}-(1-q)^{m}$ (Appendix \ref{ssec:p(q|m)}). Total error probability for $L$ clusters is bounded by $P(E) \le L\cdot P_{cl}(E)$, so:
\begin{equation}
P(E) \le L\frac{1}{K-1}\sum_{m=2}^K \frac{1-q^{m}-(1-q)^{m}}{\binom{n}{m}}.
\end{equation}
We use the well known binomial coefficient inequality $\binom{n}{m}\ge(\frac{n}{m})^{m}$ that holds for $n,m>0$ where $m<n$ and $1-q^{m}-(1-q)^{m} \le 1$, $q\in\{0,1\}$ to bound the value:
\begin{equation}
\frac{1-q^{m}-(1-q)^{m}}{\binom{n}{m}} \le \frac{1-q^{m}-(1-q)^{m}}{(\frac{n}{m})^{m}} < \frac{1}{(\frac{n}{m})^{m}},
\end{equation}
We rewrite $(\frac{n}{m})^{m}$ by using a well known inequality as $(\frac{n}{m})^{m} = (1+\frac{n-m}{m})^m \le e^{n-m}$. Plugging these expressions to the previous expression and performing simple calculations we finally obtain:
\begin{equation}
P(E) < \frac{L}{K-1} e^{2-n}\frac{e^{K-1}-1}{e-1}.
\end{equation}
For the network values $(S,L,K) = (70,5,3)$ this probability is bounded with $P(E)<1.1 \cdot 10^{-4}$.

The distance decoder error probability due to our assumption that only one defective sensor is present in the network is small. In addition, the decoder threshold value can be updated to increase the robustness. We increase the value of threshold parameter as $\epsilon^{'}=\epsilon+\delta_{\epsilon}$, where $\delta_{\epsilon} = P(E) E(\mathbf{g_1})$ and $E(\mathbf{g_1})$ is the expected number of non-zero test outcomes. It is set to the total number of observed positive test outcomes.

\section{Performance evaluation}\label{sec:exper}
\subsection{Setup}
In this section, we investigate the performance of our distributed detection method denoted as GP in various scenarii. We first examine the influence of the different network parameters in the rate of dissemination of messages. Next, we examine the decoding probability for both single and multiple defective sensor(s) detection. The number of system rounds required to collect the necessary number of messages for the accurate decoding varies with the topology. The simulations are performed for fully connected, $k$-connected and irregular graphs. 
Finally, we discuss the number of required linearly independent measurements needed for successful detection and compare it with the theoretical one. 

We also analyze the performance of several alternative schemes, namely a Random Walk method that employs a Gossip mechanism with pull protocol (RWGP) and a classical Random Walk (RW) detection. A random walk determines the path of successive random dissemination message exchanges between neighbor sensors. In the RWGP method, the random walk is initiated at $L$ sensors (equivalent to the master sensors in the GP method) and terminates after a pre-determined number of rounds. The sensors create messages from the sensor measurements collected along the random walk path. These messages are transmitted with the gossip algorithm that uses a pull protocol. Note that, for identical choice of the sensors over rounds, RWGP and GP are identical. The RW method initiates the raw (uncompressed) measurements collection in $L$ random sensors and completes it in a given number of rounds. Every sensor that lays along the random walk path stores the values of all sensors along the transmission path. When all the sensors receive all the data, the process terminates.

The GT algorithm is also compared with a Store-and-Forward (SF) and a Greedy Store-and-Forward (GSF) method that employs pull protocol. Both algorithms disseminate raw sensor measurements. For the SF method, upon receiving a message request, a node responds by forwarding randomly chosen messages from the available set of messages. In GSF, each sensor randomly requests the innovative measurements in a greedy manner from its randomly chosen neighbor sensor. This procedure involves additional message exchange among sensors in every round. 

We analyze the performance of these algorithms in fully connected, k-regular graphs and irregular networks. For irregular sensor networks construction, we place sensors randomly in a unit square area. Sensors that lay within a certain radius can communicate and exchange messages directly. In each case, we build $10$ different network realizations and for each such realization we perform $100$ independent simulations. The results are averaged over all simulations. 

\subsection{Influence of the master node selection process}
First, we study the influence of networks' capability to generate innovative messages on the decoder performance. We consider two different methods for selecting master sensors: random master sensor selection (RM) and deterministic master sensor (DM) selection. Fig. \ref{fig:20sens_randclusterno1} illustrates the detection probability and the achieved average rank with respect to the number of message dissemination rounds, for fully connected graphs with $S=20$ sensors and one ($K=1$) defective sensor. We observe that the performance depends on $L$ and ${\alpha}=qK$ for both RM and DM. These values should be selected properly in order to maximize the information diversity in the network. Specifically, we observe that RM achieves the maximum message diversity for $\alpha=1$ (maximum value) since the diversity of messages in this case is maximized by construction in Fig. \ref{fig:20sens_randclusterno1}. We can also note that the number of clusters does not affect significantly the detection performance of RM. On the contrary, for DM both parameters $L$ and $\alpha$ are important. Small values of $\alpha$ guarantee high message diversity. This is due to the fact that DM requires more rounds to receive enough messages for detection. In the following, we focus on RM selection where possible (that is, for $K=1$), as it provides higher probability of creating innovative messages. 
\begin{figure*}[thb]
\begin{center}
\begin{tabular}{cc}
~\includegraphics[width=  7cm]{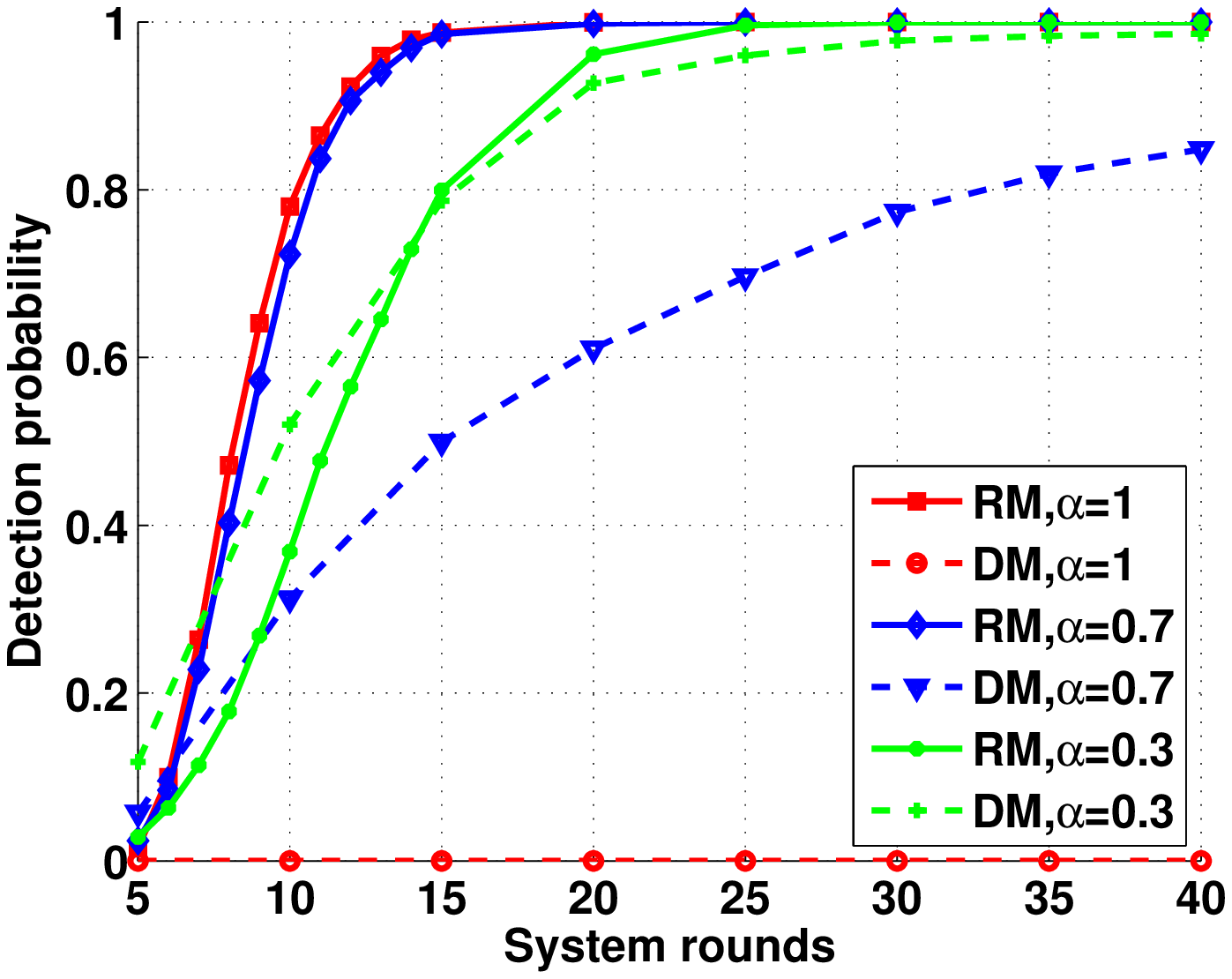}~&
~\includegraphics[width=  7cm]{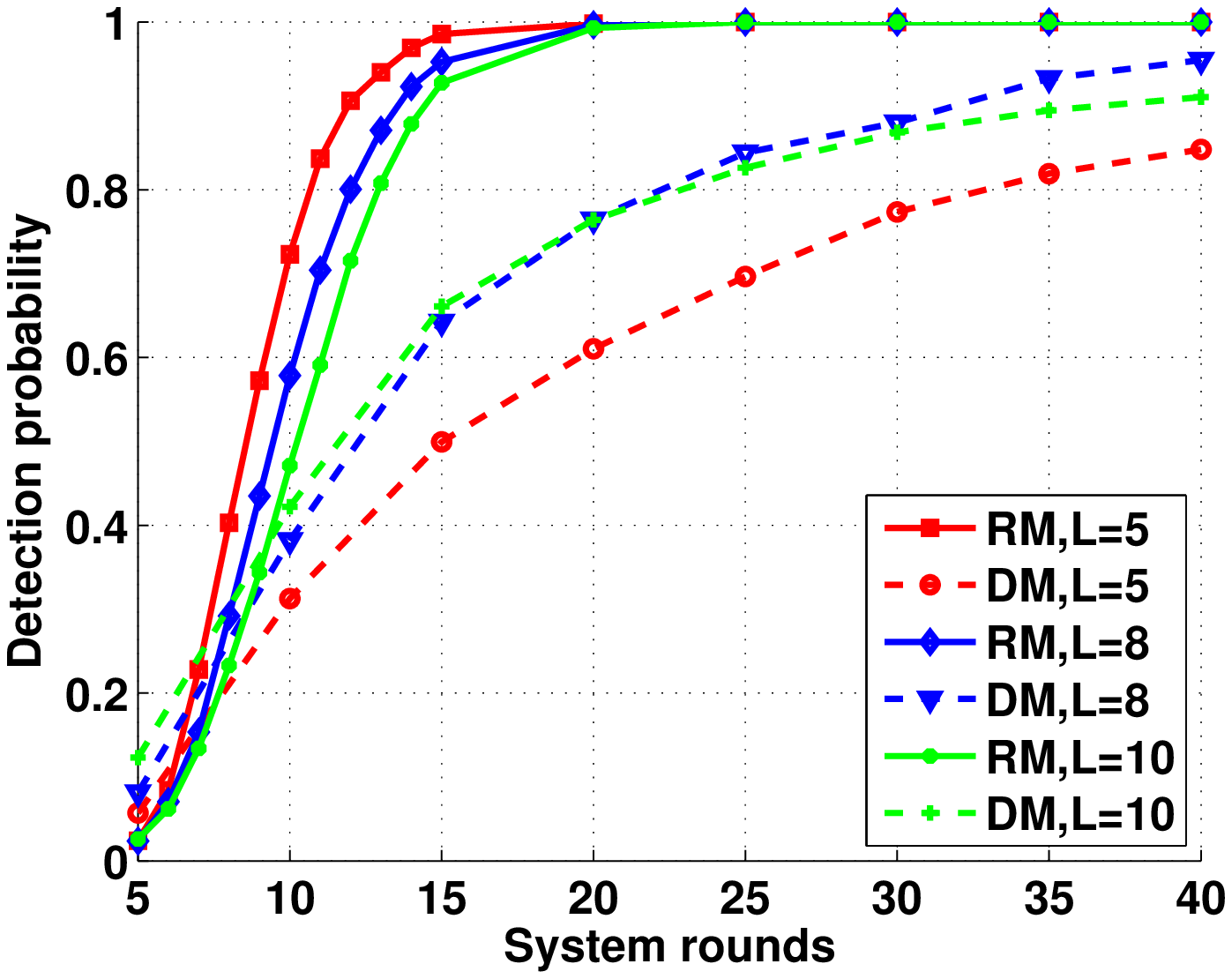}~\\
~\includegraphics[width=  7cm]{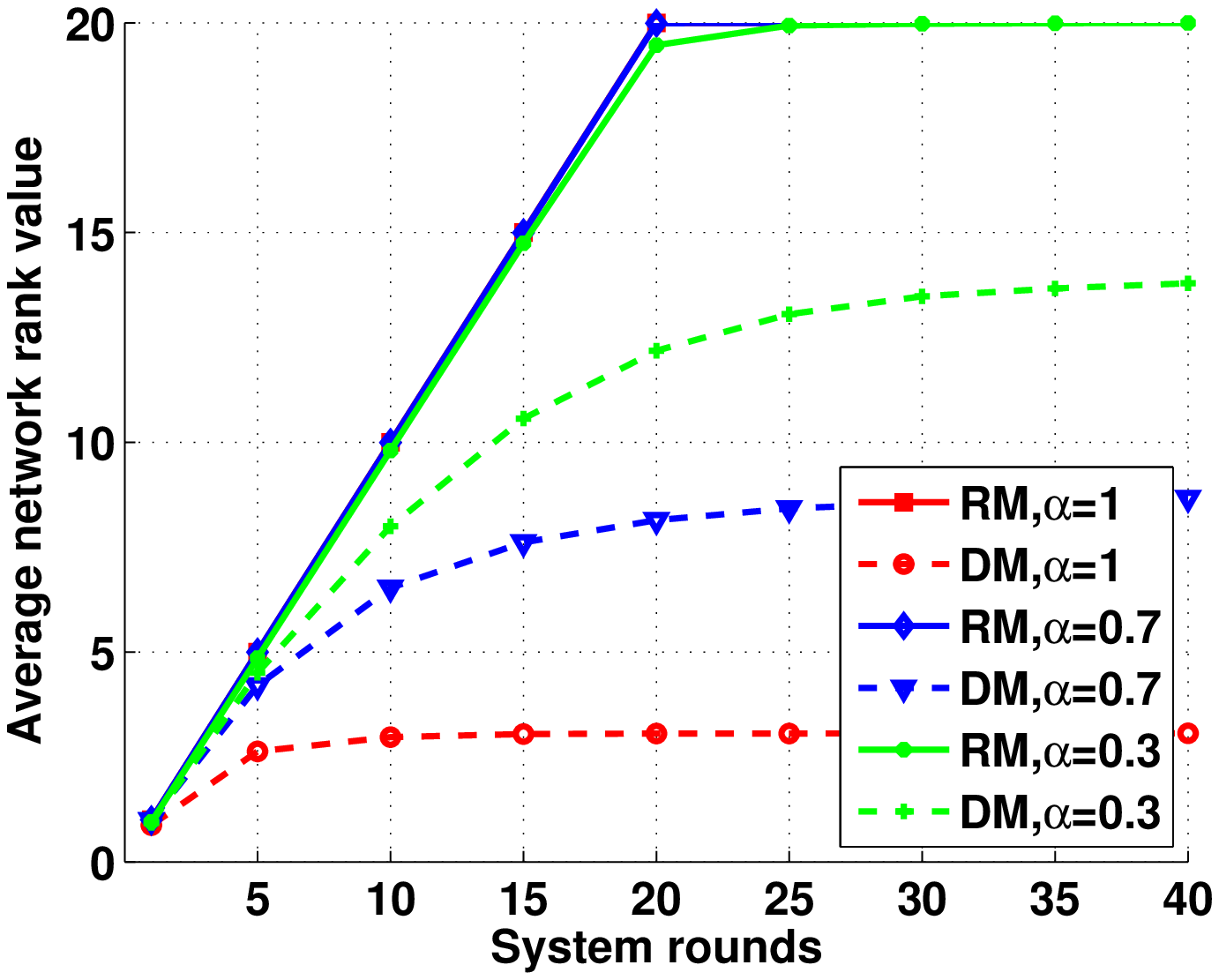}~&
~\includegraphics[width=  7cm]{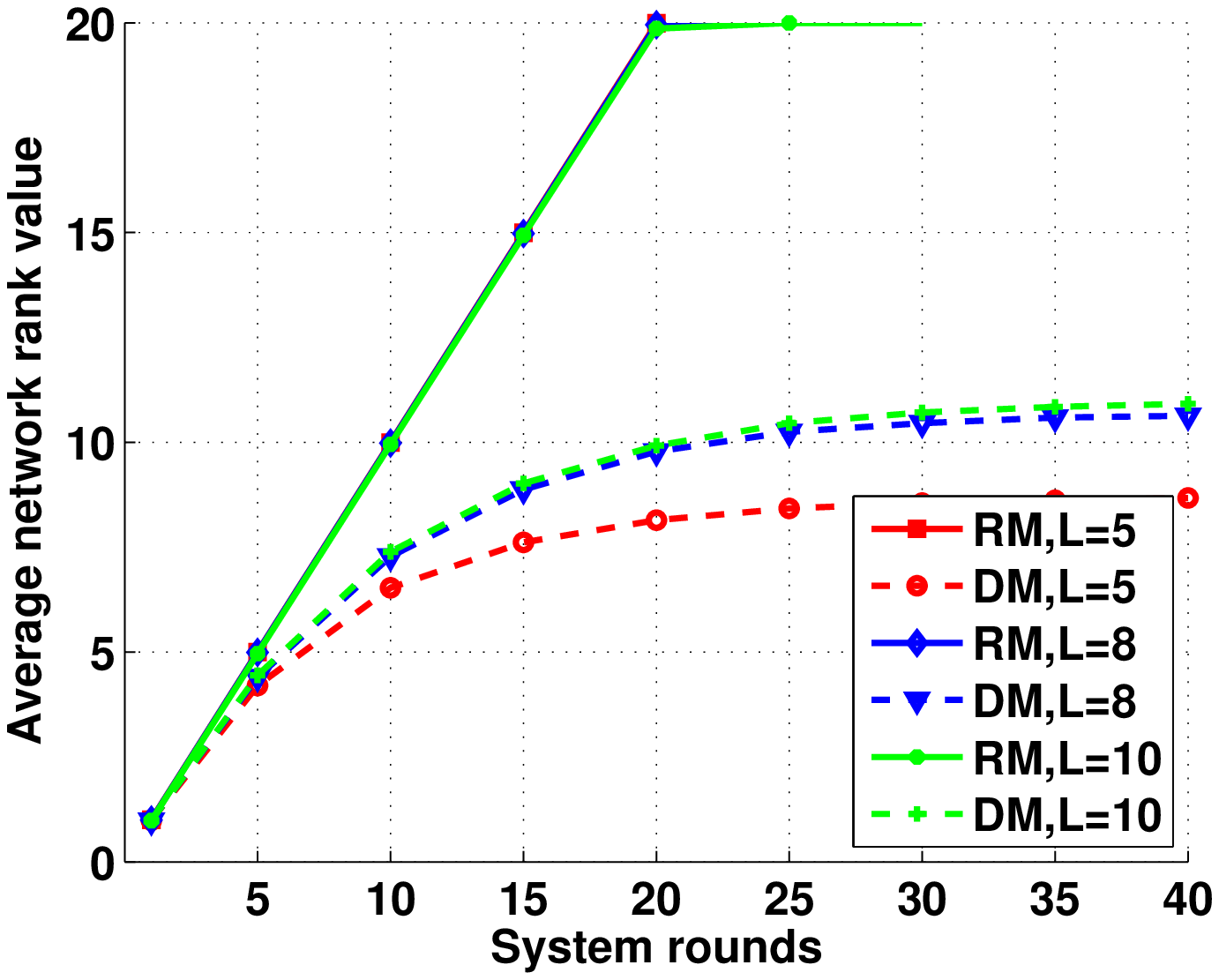}~\\
~(a) $L=5$~&~(b) $\alpha=0.7$~\\ 
\end{tabular}
\end{center}
\caption{Simulation results for fully connected graphs with $S=20$ sensors, $K=1$, where RM and DM denote the random and deterministic selection mode of master sensors, respectively. Top row: Probability of defective sensor detection. Bottom row: Average rank of messages received per sensor. Column (a): fixed values of the master sensors ($L=5$). Column (b): fixed values of the sensor participation constant ($\alpha = qK= 0.7$).}
\label{fig:20sens_randclusterno1}
\end{figure*}
\subsection{Detection performance}
We first consider the case of a single defective sensor ($K=1$). The detection probability and the average rank evolution over rounds are examined for fully connected (FG) and $k$-connected regular networks (RG) with sensors degree $k\in\{6,16\}$. For all cases, the network consists of $S=20$ sensors. From Fig. \ref{fig:20sens_k6_conn} we see that networks with higher number of connections achieve faster dissemination of innovative messages. We also note that high connectivity value $k$ is beneficial, but it cannot drive by itself the performance of our detection scheme. It should be combined with appropriate choice of network parameters, as discussed earlier. For example, RM master sensor selection for $k=16$ achieves better detection performance, compared to that of fully connected graphs.

\begin{figure*}[thb]
\begin{center}
\begin{tabular}{cc}
~\includegraphics[width=  7cm]{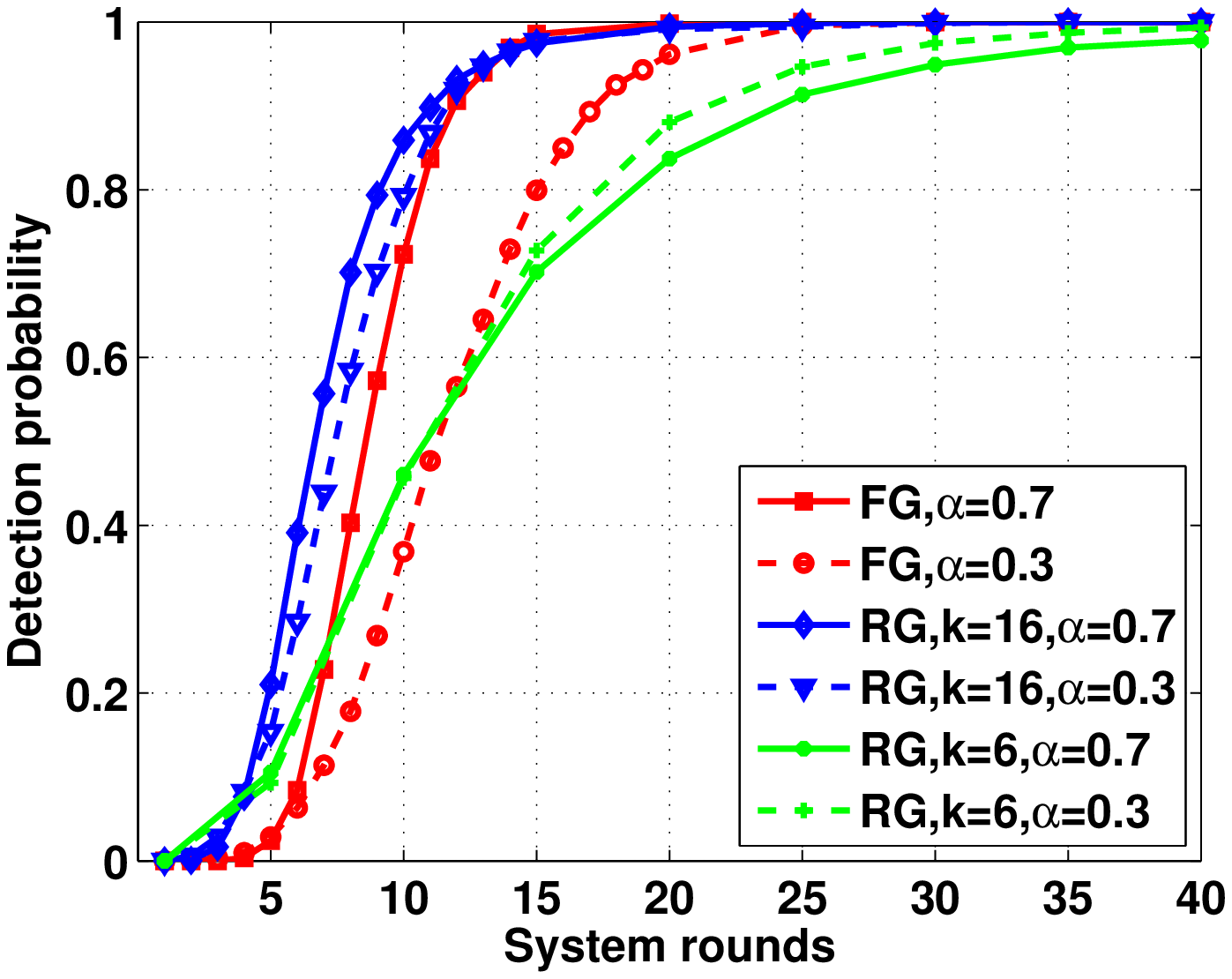}~&
~\includegraphics[width=  7cm]{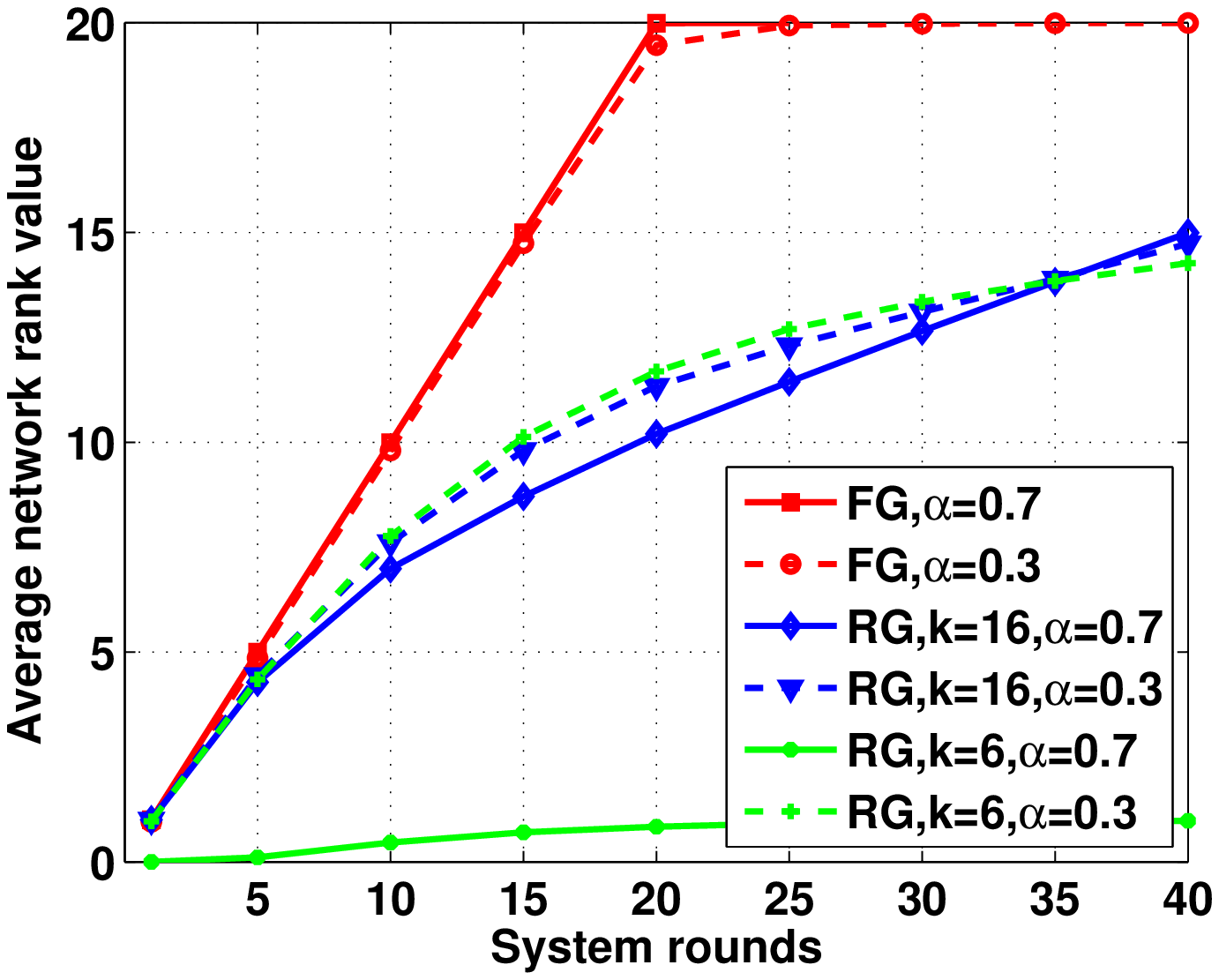}~\\
~(a)~&~(b) \\
\end{tabular}
\end{center}
\caption{Simulation results for fully connected (FG), $k=16$-regular connected (RG, $k=16$) and $k=6$-connected graphs (RG, $k=6$) with $S=20$ sensors, $K=1$ and a random selection (RM) of $L=5$ master sensors: (a) Probability of defective sensor detection; (b) Average rank of messages received per sensor.}
\label{fig:20sens_k6_conn}
\end{figure*}

In Fig. \ref{fig:20sens_irregular1_alpha0_7}, we illustrate the detection probability for random graphs ($100$ simulations per different graph) with $S=20$, $K=1$ defective sensor, $L=5$ random clusters and minimum sensors' degree $k\ge3$. We observe that random graphs require more rounds in average for successful detection, as expected. Also, we observe that the detection performance decreases because of the limited message diversity (smaller probability of receiving innovative messages) and the low connectivity. Similarly, Fig. \ref{fig:70sens_irregular1_alpha0_7} presents results for larger networks which are in accordance with the above.

\begin{figure*}[thb]
\begin{center}
\begin{tabular}{cc}
~\includegraphics[width=  7cm]{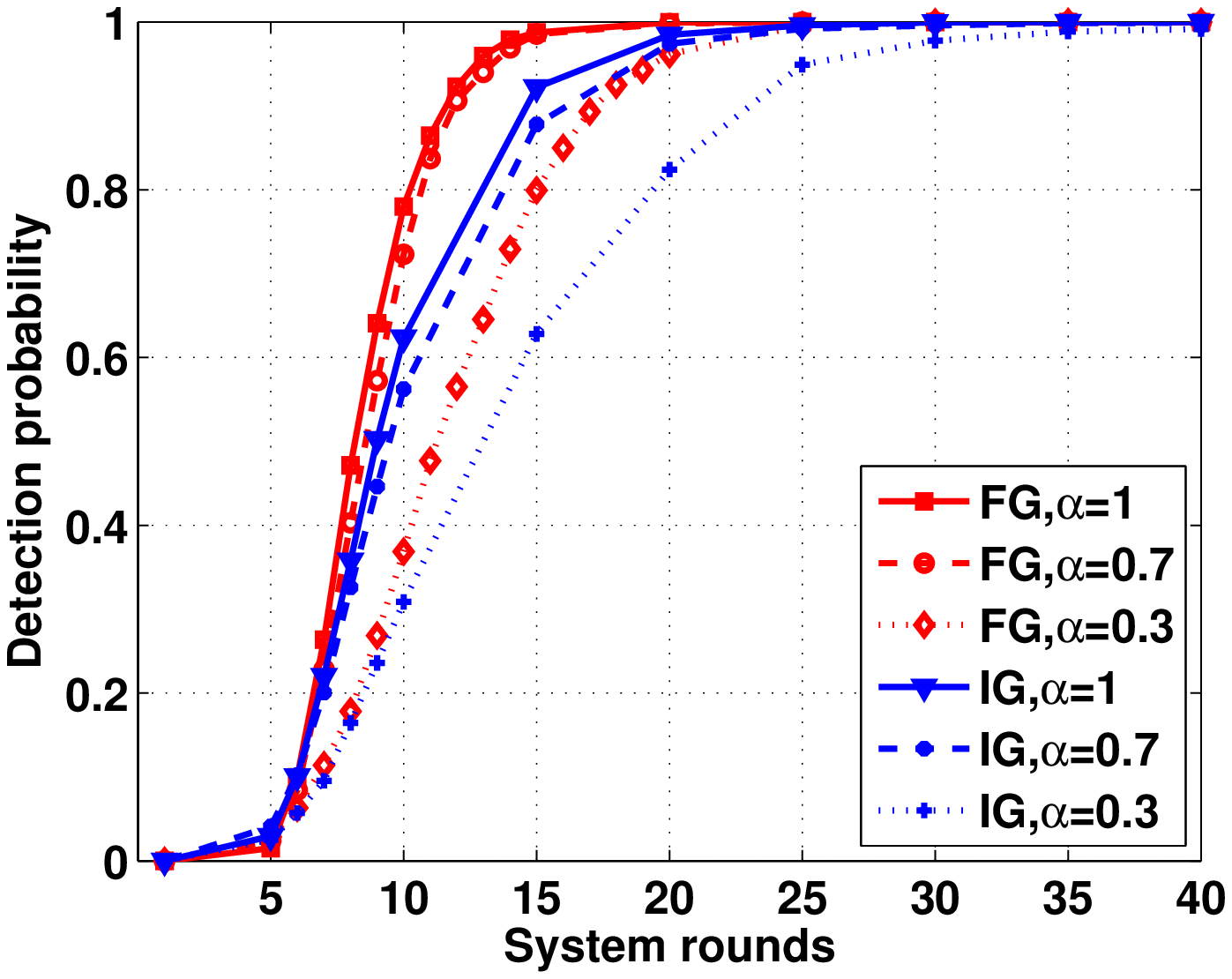}~&
~\includegraphics[width=  7cm]{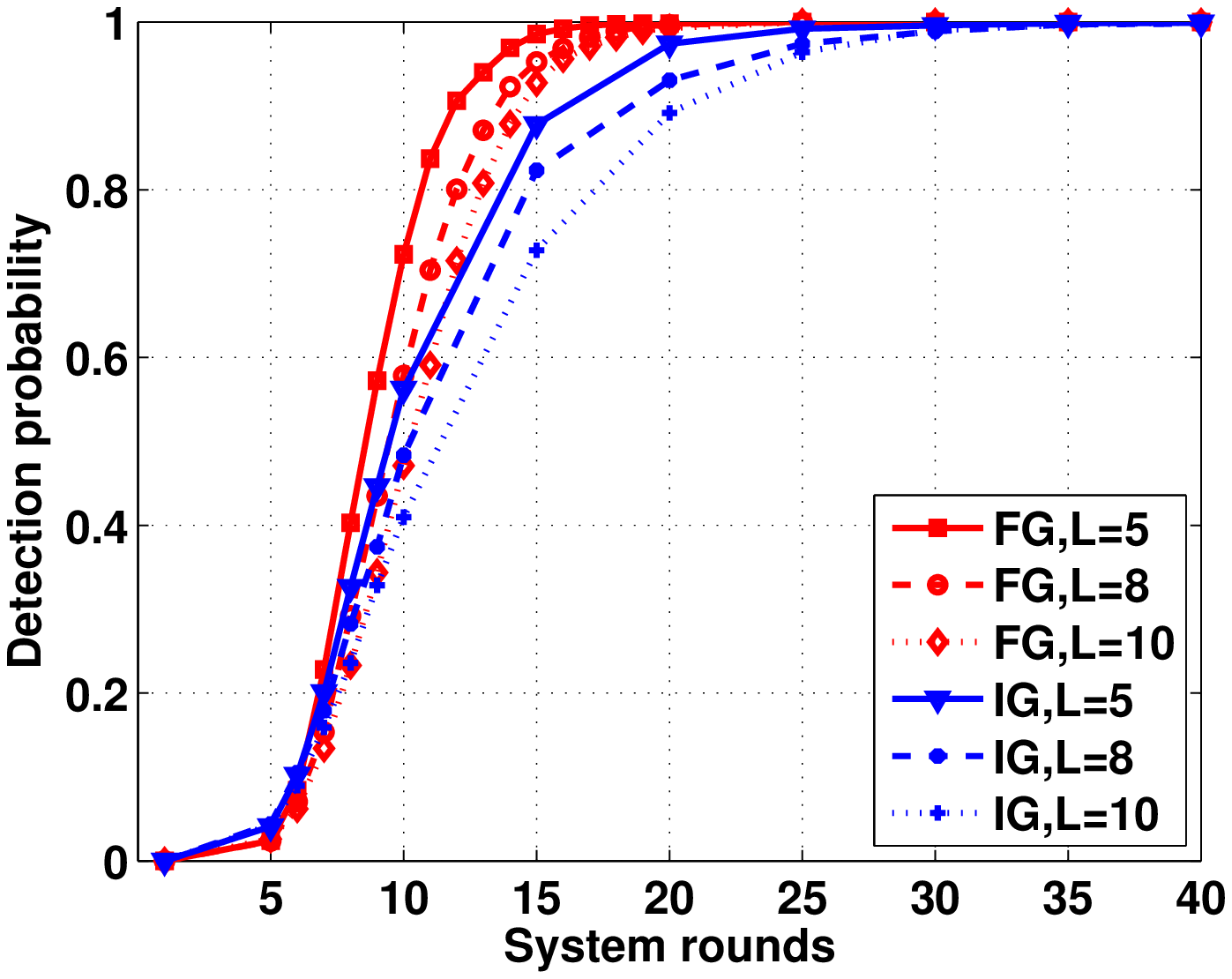}~\\
~(a) $L=5$~&~(b)$\alpha=0.7$~\\
\end{tabular}
\end{center}
\caption{Probability of defective sensor detection; Simulation results for irregular graphs ($k>3$) and random selection (RM) of $S=20$ sensors, $K=1$. (a) $L=5$ master sensors; (b) sensor participation constant $\alpha=qK=0.7$.}
\label{fig:20sens_irregular1_alpha0_7}
\end{figure*}
\begin{figure*}[htb]
\begin{center}
\begin{tabular}{cc}
~\includegraphics[width=  7cm]{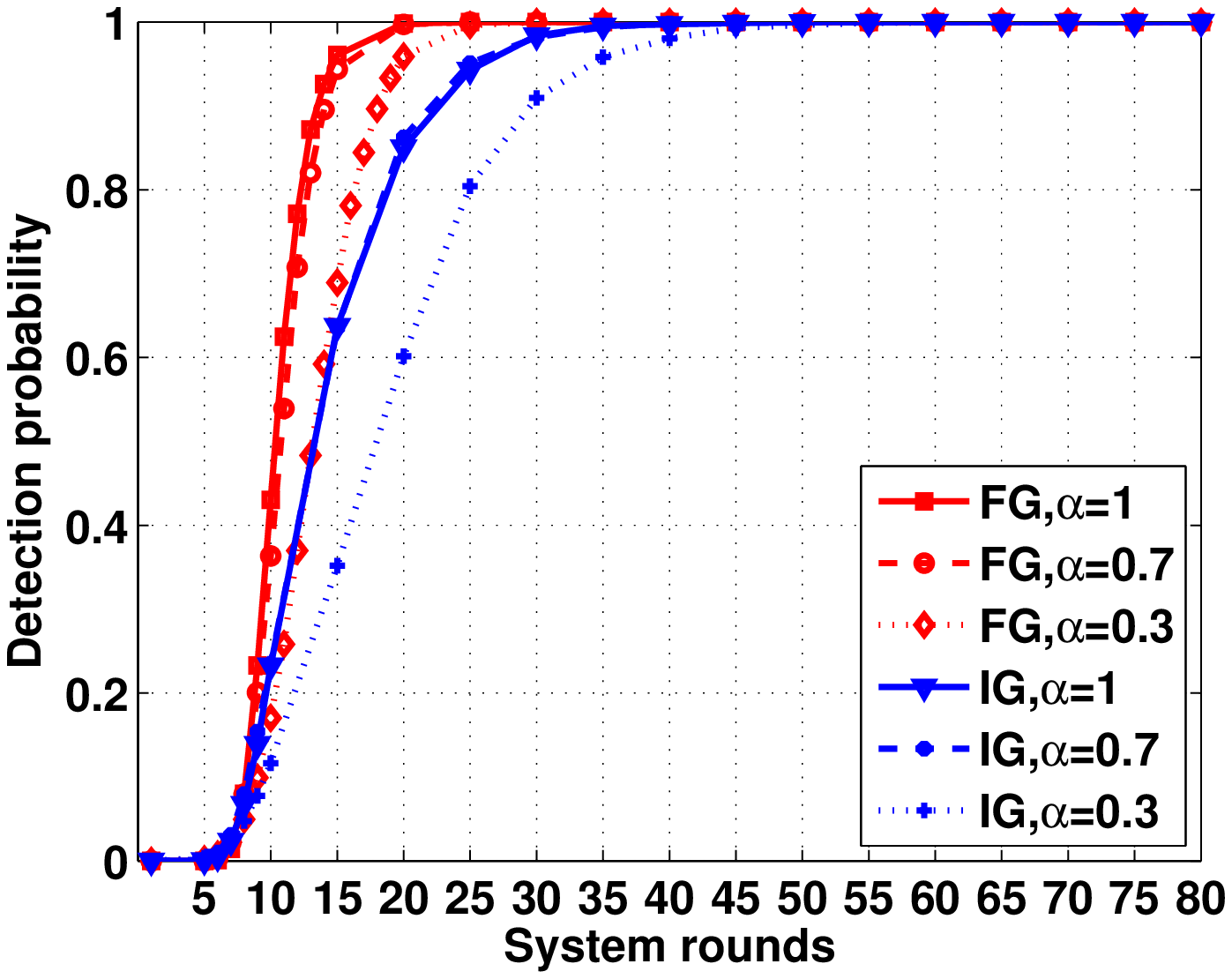}~&
~\includegraphics[width=  7cm]{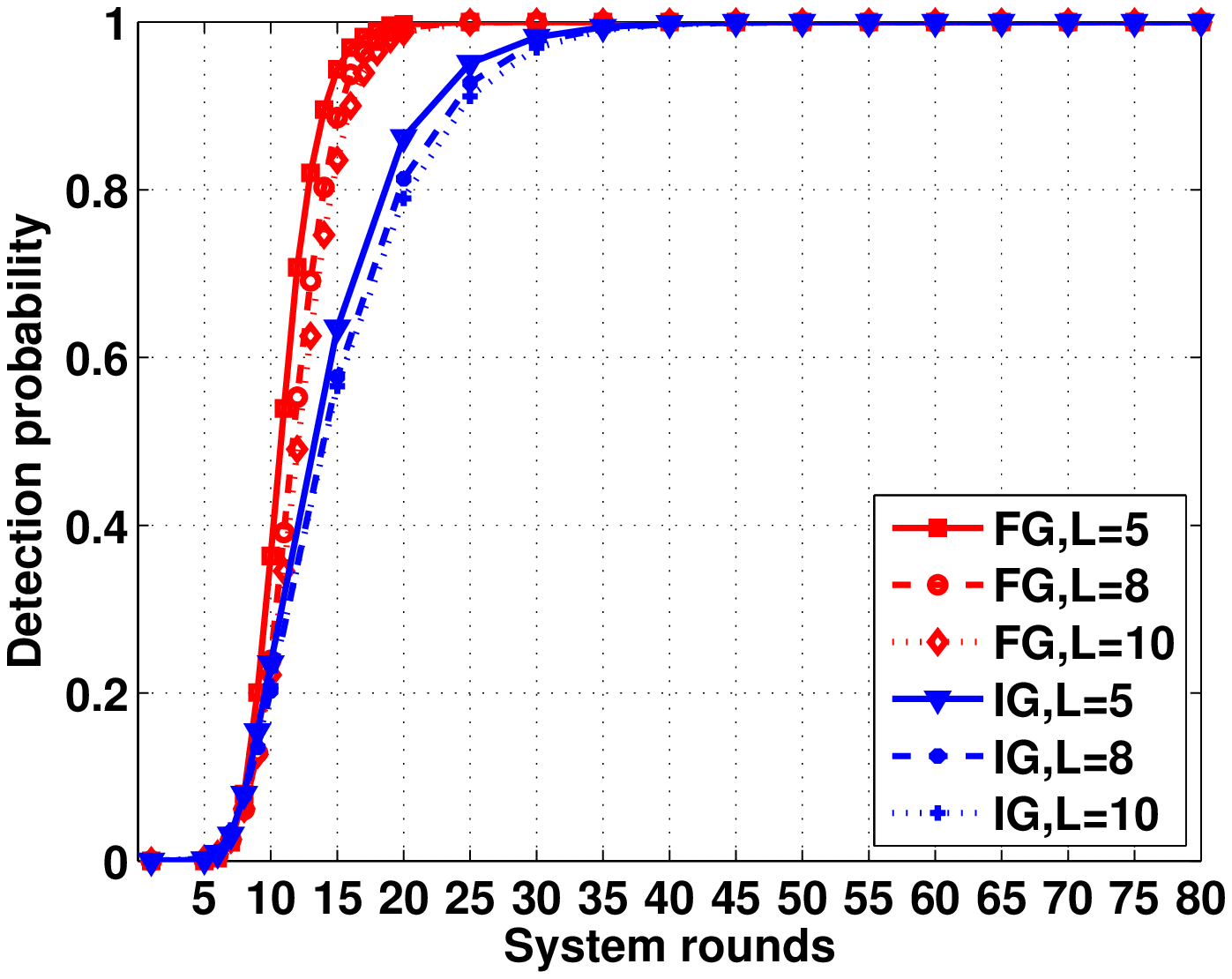}~\\
~(a)$L=5$ ~&~(b)$\alpha=0.7$~\\
\end{tabular}
\end{center}
\caption{Probability of defective sensor detection; Simulation results for irregular graphs ($k>3$) and random selection (RM) of $S=70$ sensors, $K=1$. (a) $L=5$ master sensors; (b) sensor participation constant $\alpha=qK=0.7$.}
\label{fig:70sens_irregular1_alpha0_7}
\end{figure*}
We then consider the case of multiple defective sensors. In Figs. \ref{fig:20sens_irregular_L10_K2} and \ref{fig:20sens_irregular1_alpha0_7_K2} we present results for the cases with two defective sensors ($K=2$) in networks of $20$ sensors. The results are given in terms of the average detection probability over dissemination rounds, for both fully and irregularly connected graphs. The master sensors are selected deterministically (DM) due to decoder design for multiple defective sensors identification. Note that this example violates the condition $K\ll S$ and the performance of the detection algorithm is pretty poor. In addition, results for $S=70$ and $K=2$ are depicted in Figs. \ref{fig:70sens_irregular_L10_K2} and \ref{fig:70sens_irregular1_alpha0_3_K2}. We focus on the evolution of the decoding probability and the average number of messages collected over rounds. From the evaluation it is clear that the detection performance is reasonable when the selected parameters value $(L,\alpha)$ favor diverse message generation. 
\begin{figure*}[htb]
\begin{center}
\begin{tabular}{cc}
~\includegraphics[width=  7cm]{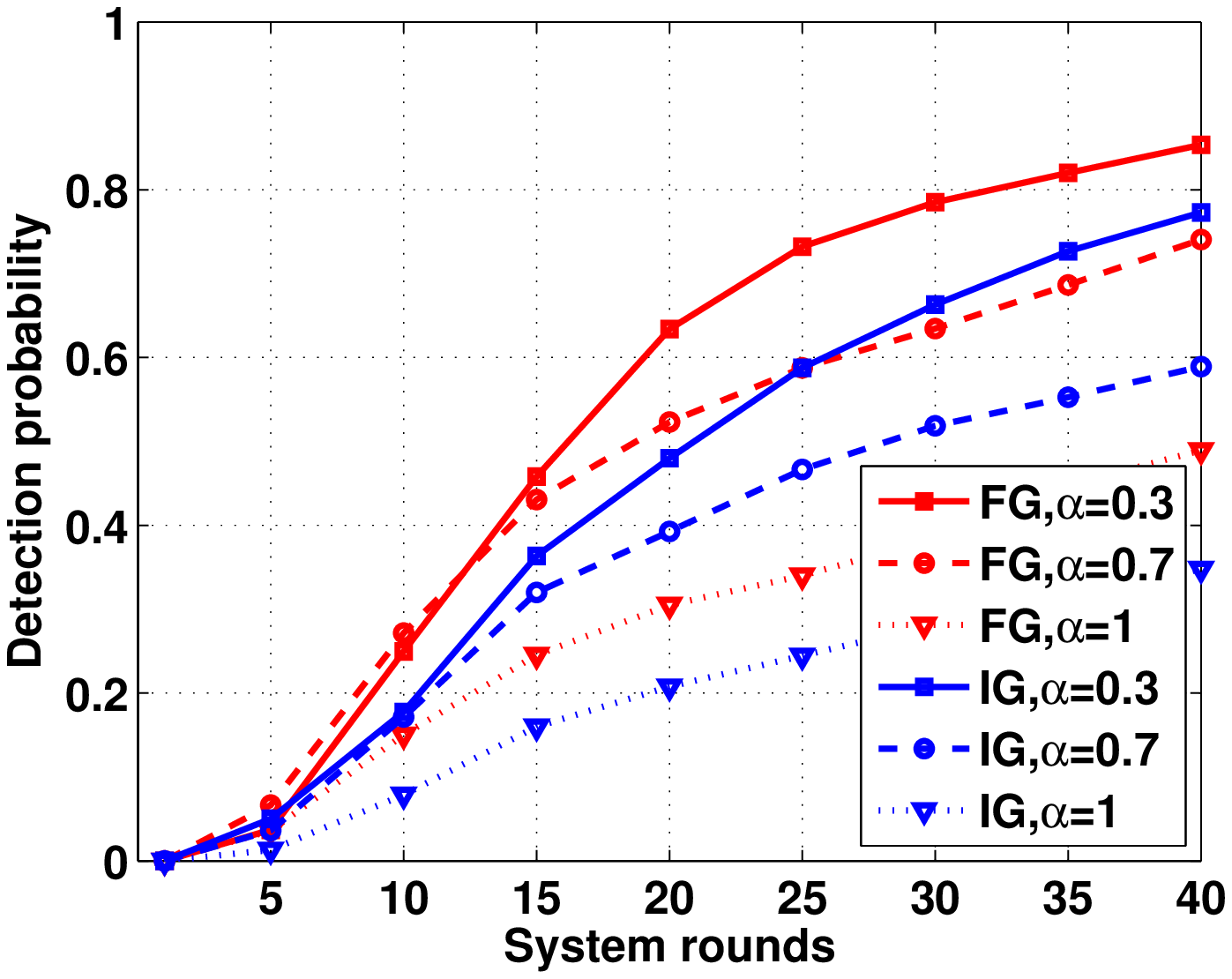}~&
~\includegraphics[width=  7cm]{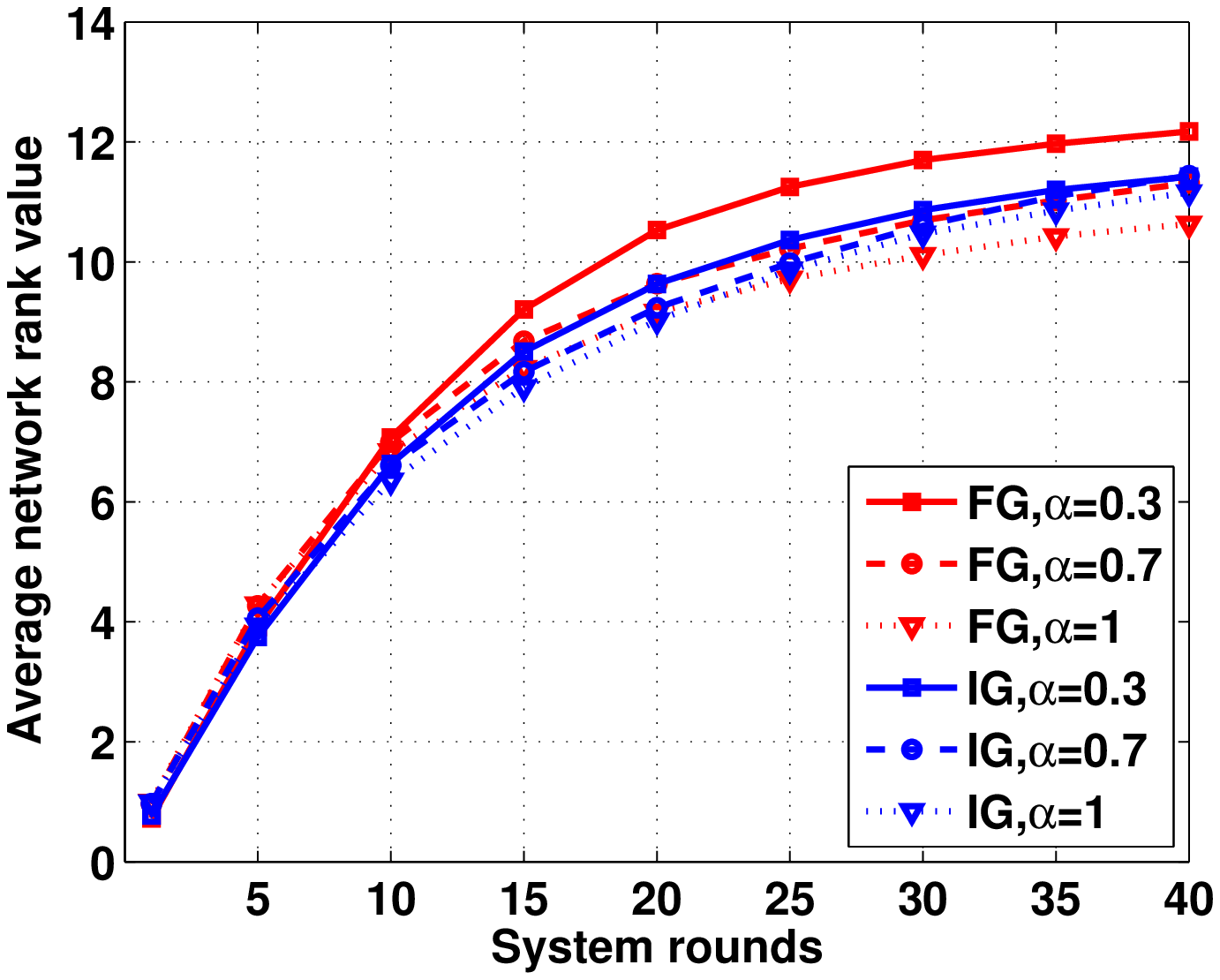}~\\
~(a)~&~(b)~\\\end{tabular}
\end{center}
\caption{Simulation results for fully connected (FG) and irregular graphs (IG), $d>3$ with $S=20$ sensors, $K=2$ and deterministic selection (DM) of $L=5$ master sensors: (a) Probability of defective sensor detection; (b) Average rank value.}
\label{fig:20sens_irregular_L10_K2}
\end{figure*}
\begin{figure*}[htb]
\begin{center}
\begin{tabular}{cc}
~\includegraphics[width=  7cm]{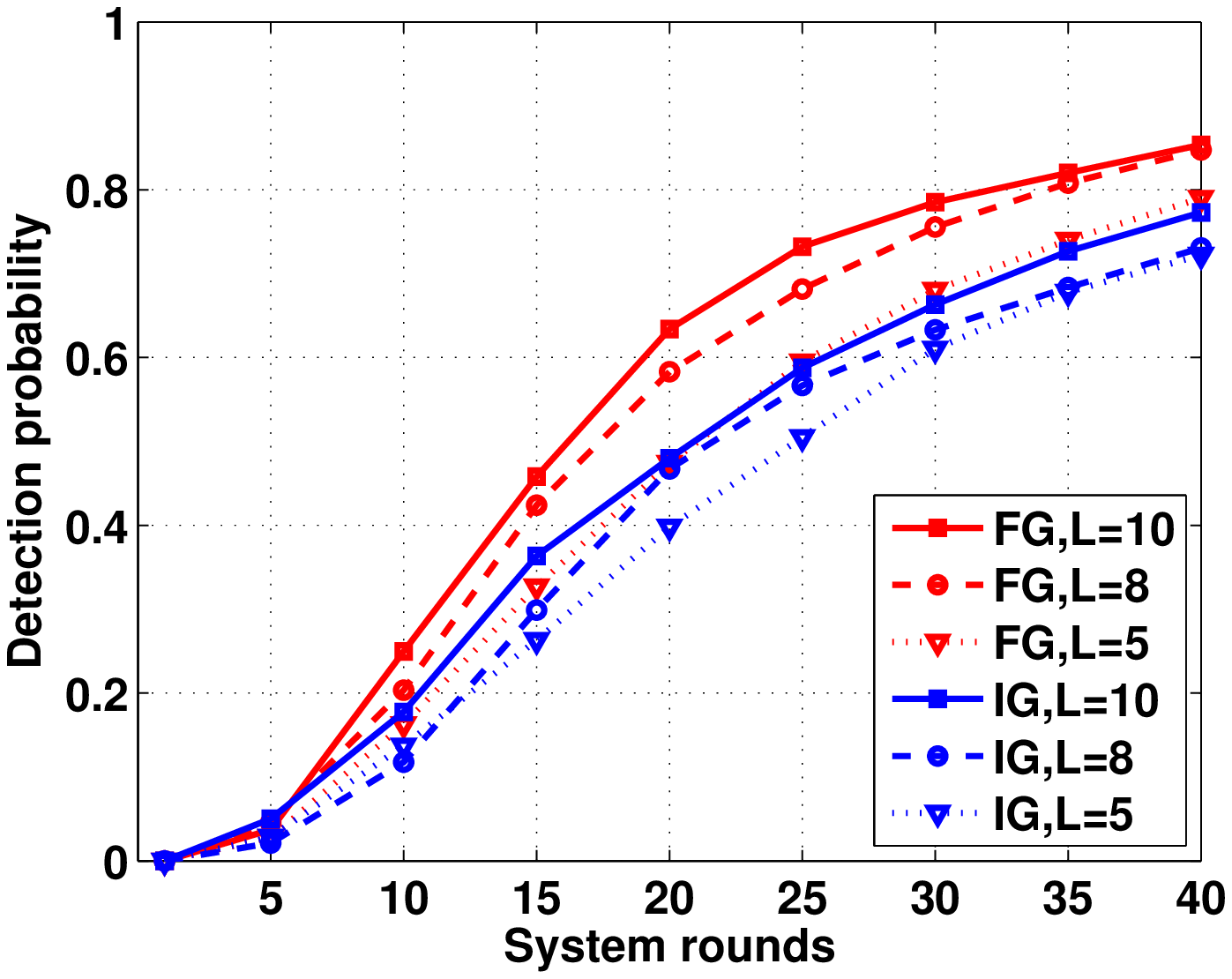}~&
~\includegraphics[width=  7cm]{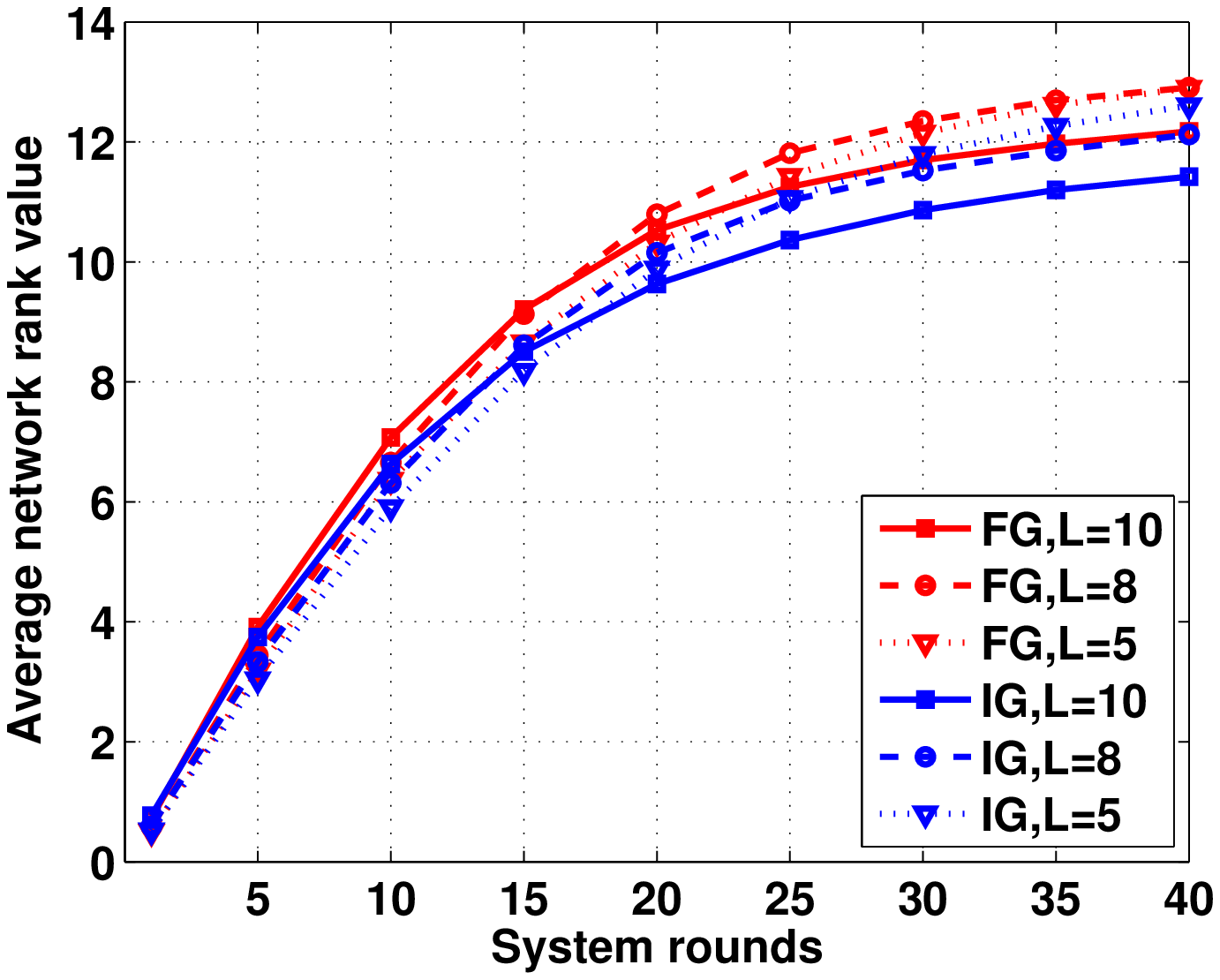}~\\
~(a)~&~(b)~\\\end{tabular}
\end{center}
\caption{Simulation results for fully connected (FG) and irregular graphs (IG), $d>3$ with $S=20$ sensors, $K=2$ and deterministic selection (DM) of master sensors, $\alpha=0.3$: (a) Probability of defective sensor detection; (b) Average rank value.}
\label{fig:20sens_irregular1_alpha0_7_K2}
\end{figure*}
\begin{figure*}[htb]
\begin{center}
\begin{tabular}{cc}
~\includegraphics[width=  7cm]{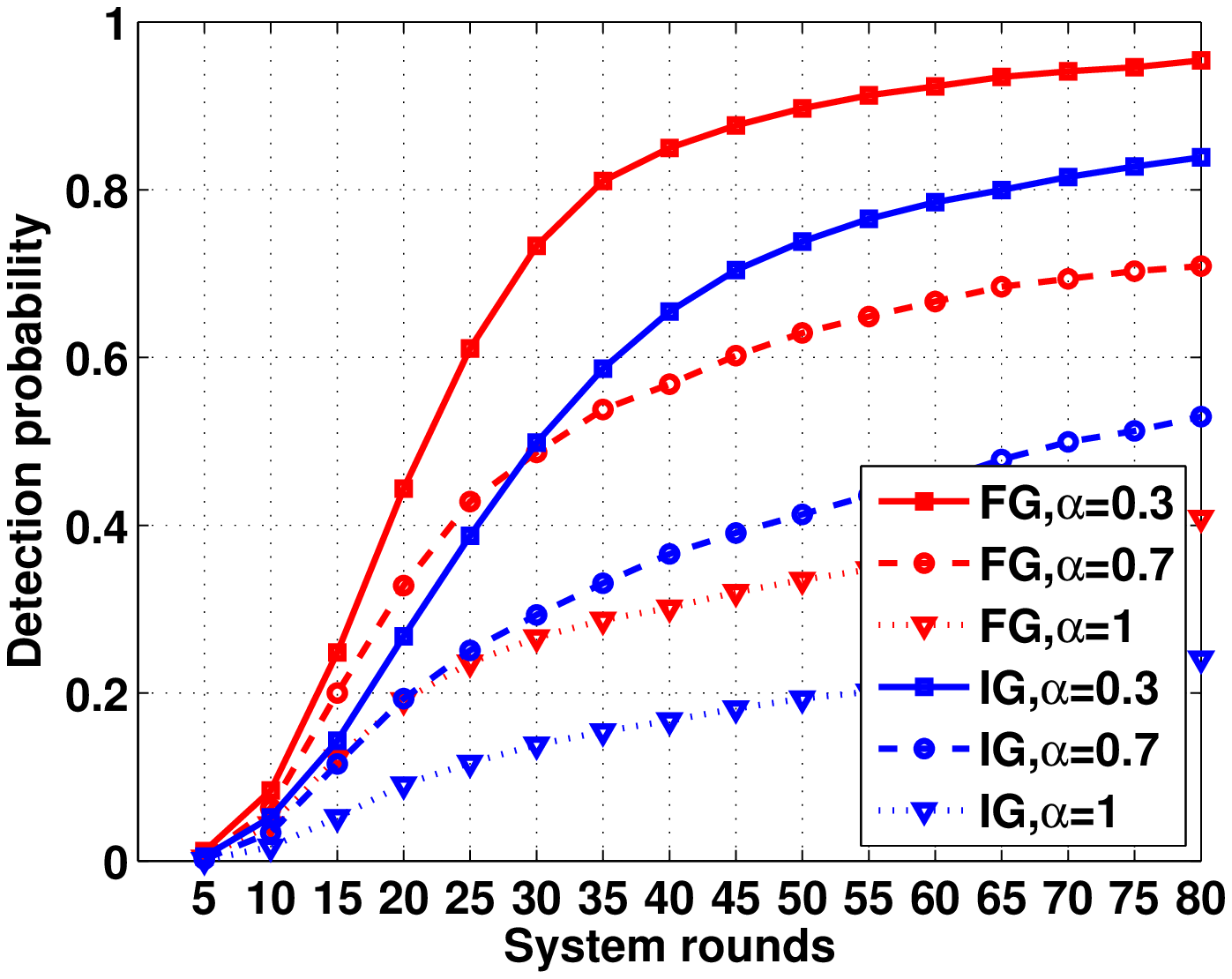}~&
~\includegraphics[width=  7cm]{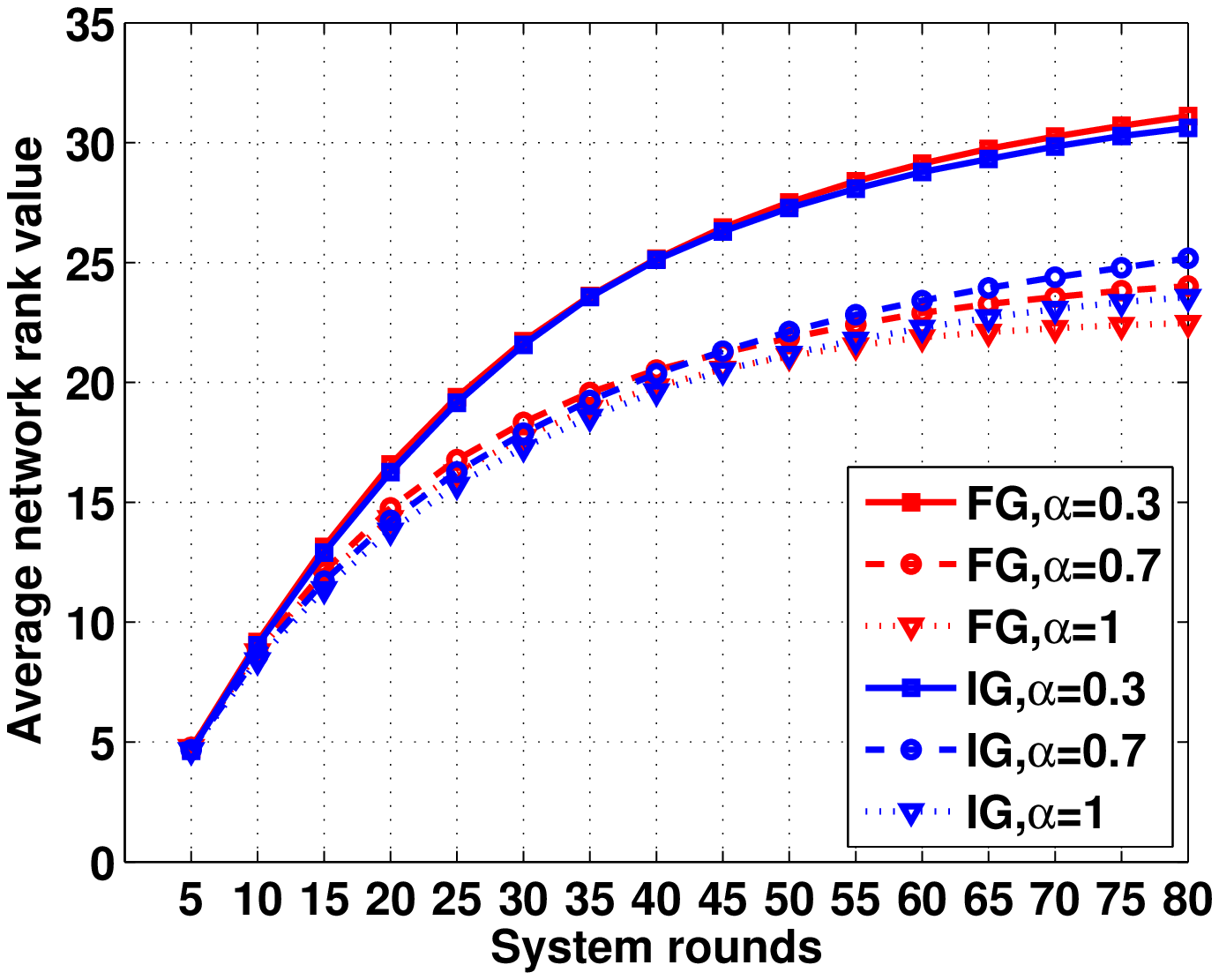}~\\
~(a)~&~(b)~\\\end{tabular}
\end{center}
\caption{Simulation results for fully connected (FG) and irregular graphs (IG), $d>3$ with $S=70$ sensors, $K=2$ and deterministic selection (DM) of $L=10$ master sensors: (a) Probability of defective sensor detection (b) Average rank value.}
\label{fig:70sens_irregular_L10_K2}
\end{figure*}
\begin{figure*}[htb]
\begin{center}
\begin{tabular}{cc}
~\includegraphics[width=  7cm]{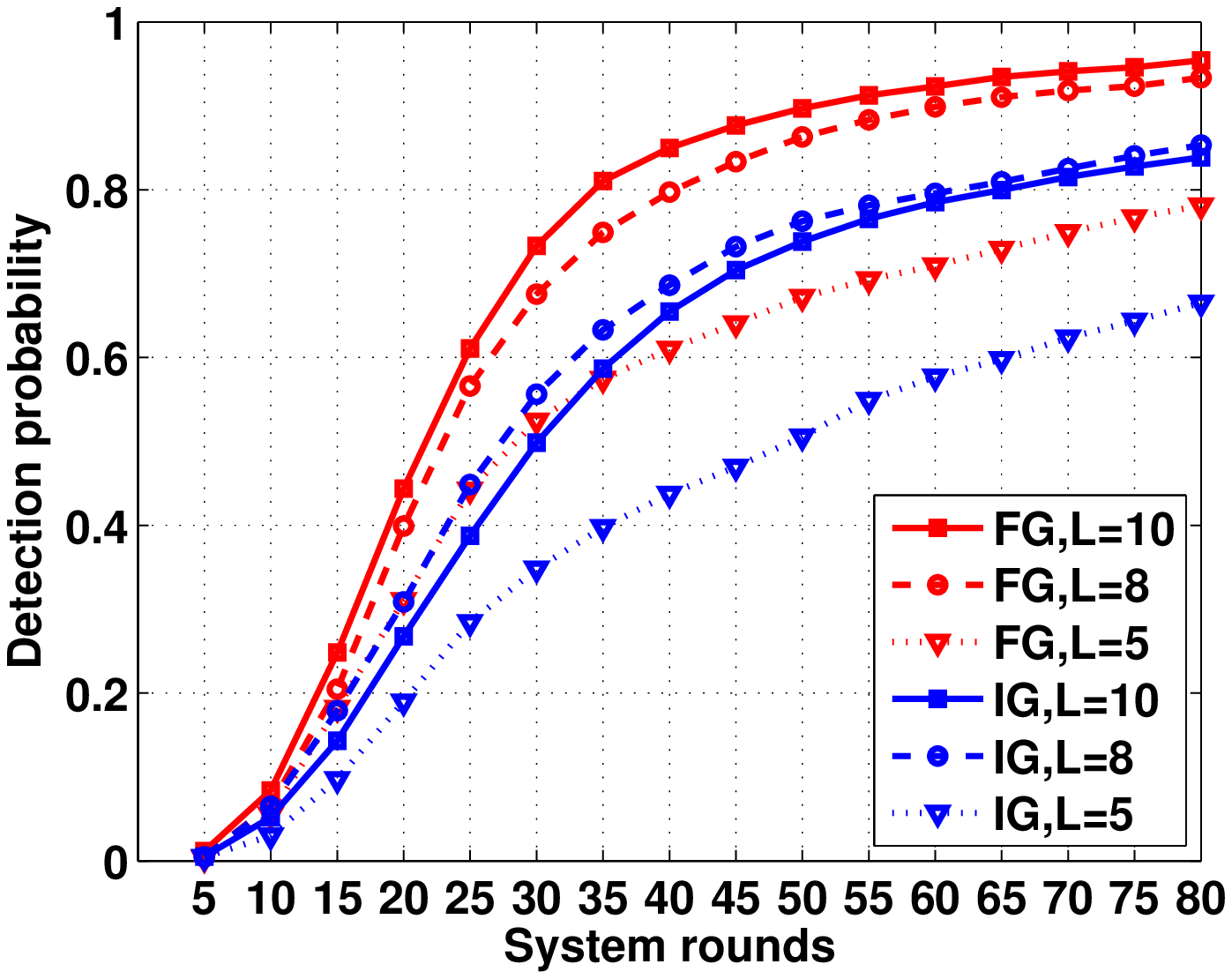}~&
~\includegraphics[width=  7cm]{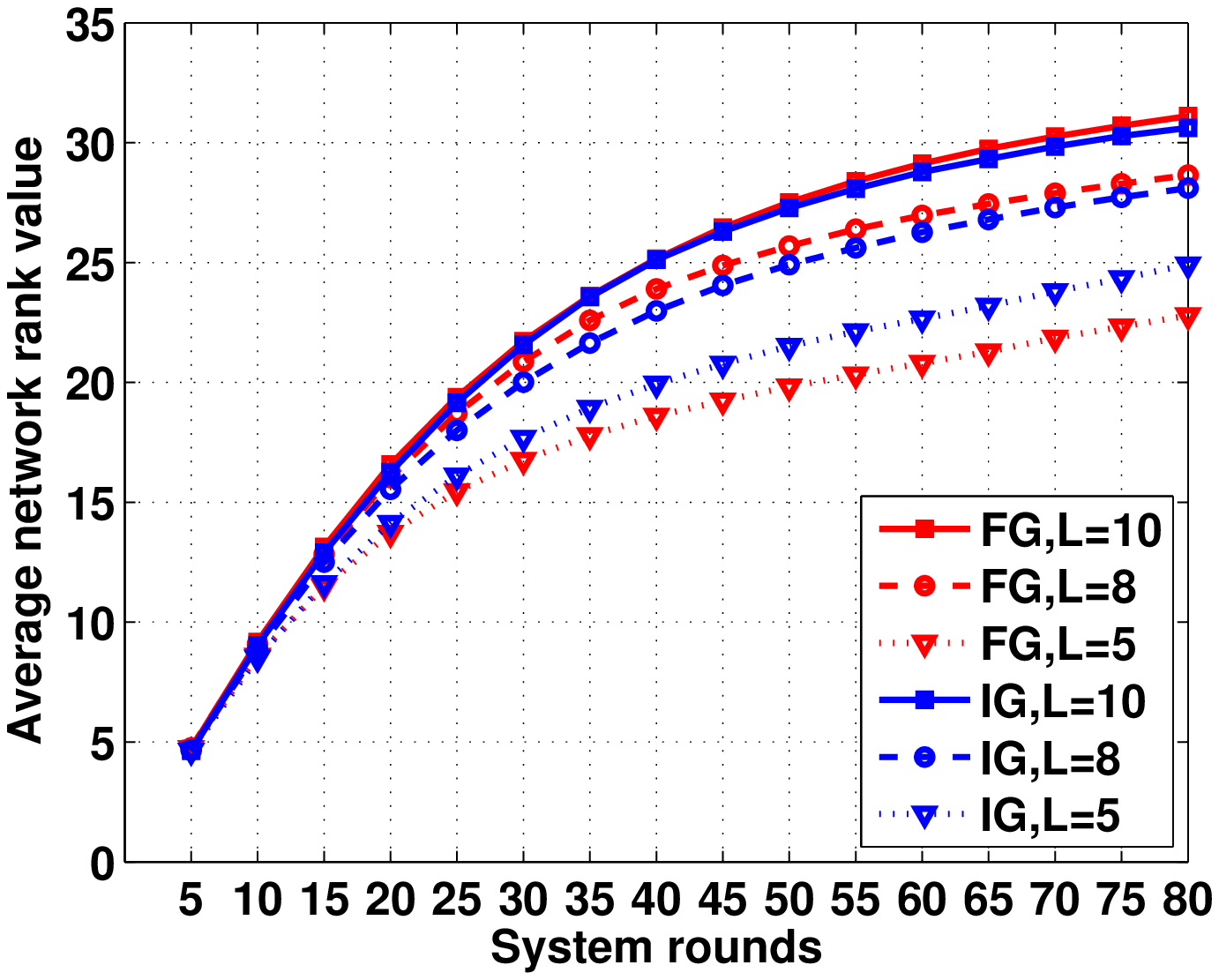}~\\
~(a)~&~(b)~\\\end{tabular}
\end{center}
\caption{Simulation results for fully connected (FG) and irregular graphs (IG), $d>3$ with $S=70$ sensors, $K=2$ and deterministic selection (DM) of master sensors, $\alpha=0.3$: (a) Probability of defective sensor detection; (b) Average rank value.}
\label{fig:70sens_irregular1_alpha0_3_K2}
\end{figure*}

In \cite{Cheraghchi:11}, a centralized system has been proposed, which can be considered as dual to fully connected networks with centralized tests (single master sensor that covers all the network). For comparison reasons, we compute the required number of measurements for networks with: $(S=20,K\in\{1,2\},p\in(0.9-1), q\in(0.15-0.3),pf_1=0.01,pf_2=0.01)$ and $(S=70,K\in\{1,2\},p\in(0.9-1), q\in(0.15-0.3),pf_1=0.01,pf_2=0.01)$. The results are reported in Table \ref{tab:theory_requir}. We observe that the worst case analysis leads to higher number of dissemination rounds than the real ones. However, these values decrease relatively to the growth of number of sensors in the network. Simulations show that in practice the required measurements are significantly fewer.
\begin{table}
\begin{center}
\caption{The theoretical measurement requirements for networks with $S$ sensors.}
\begin{tabular}{|c|c|c|c|c|}
\hline
&\multicolumn{2}{|c|}{S=20} &\multicolumn{2}{|c|}{S=70}\\
&$K=1$&$K=2$&$K=1$&$K=2$\\
\hline
$p\in(0.9-1)$&~130  & (115-244)& (174-217)&(125-284)\\
\hline
\end{tabular}\label{tab:theory_requir}
\end{center}
\end{table}

Detection probability comparison of the proposed method with several detection methods are illustrated in Figs. \ref{fig:20sens_dissemination} and \ref{fig:70sens_dissemination}, for $20$ and $70$ sensors respectively. The proposed scheme outperforms all other methods. Note that the number of necessary rounds in RWGP scheme is large compared to the other schemes, while RW needs higher communication overhead for dissemination due to the transmission of raw sensor measurements. Average rank values over the network rounds are illustrated in Fig. \ref{fig:sens_dissemination_rank}.
We observe that for the fixed detection probability $p=0.9$ for the network with $S\in \{20,70\}$ sensors the average number of system rounds required for the proposed method is approximately $\{13,17\}$ and $\{15,20\}$, respectively. The number of system rounds required by the other algorithms to reach the same probability of performance is higher, especially for the network with $70$ sensors. 

\begin{figure*}[htb]
\begin{center}
\begin{tabular}{cc}
~\includegraphics[width=  7cm]{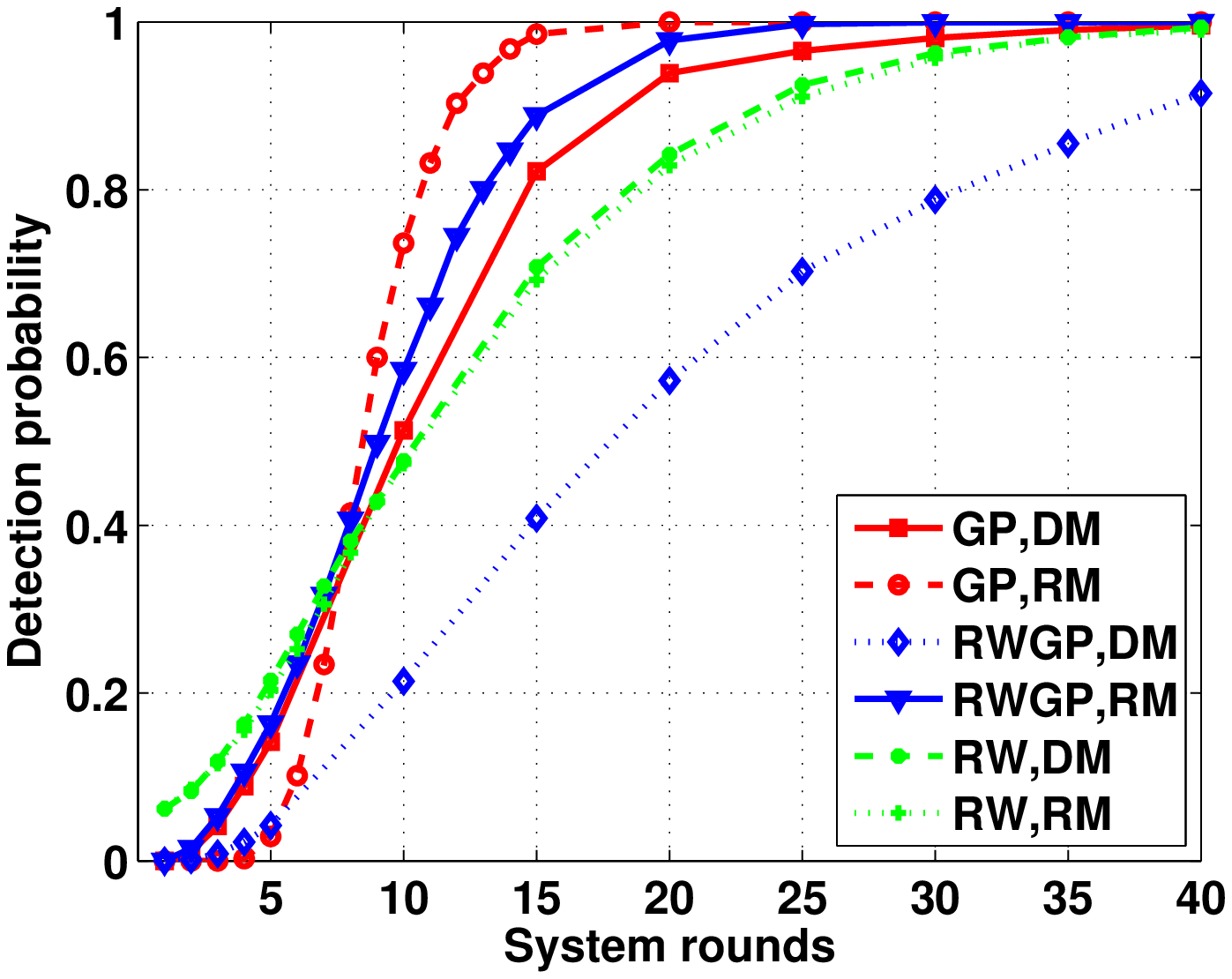}~&
~\includegraphics[width=  7cm]{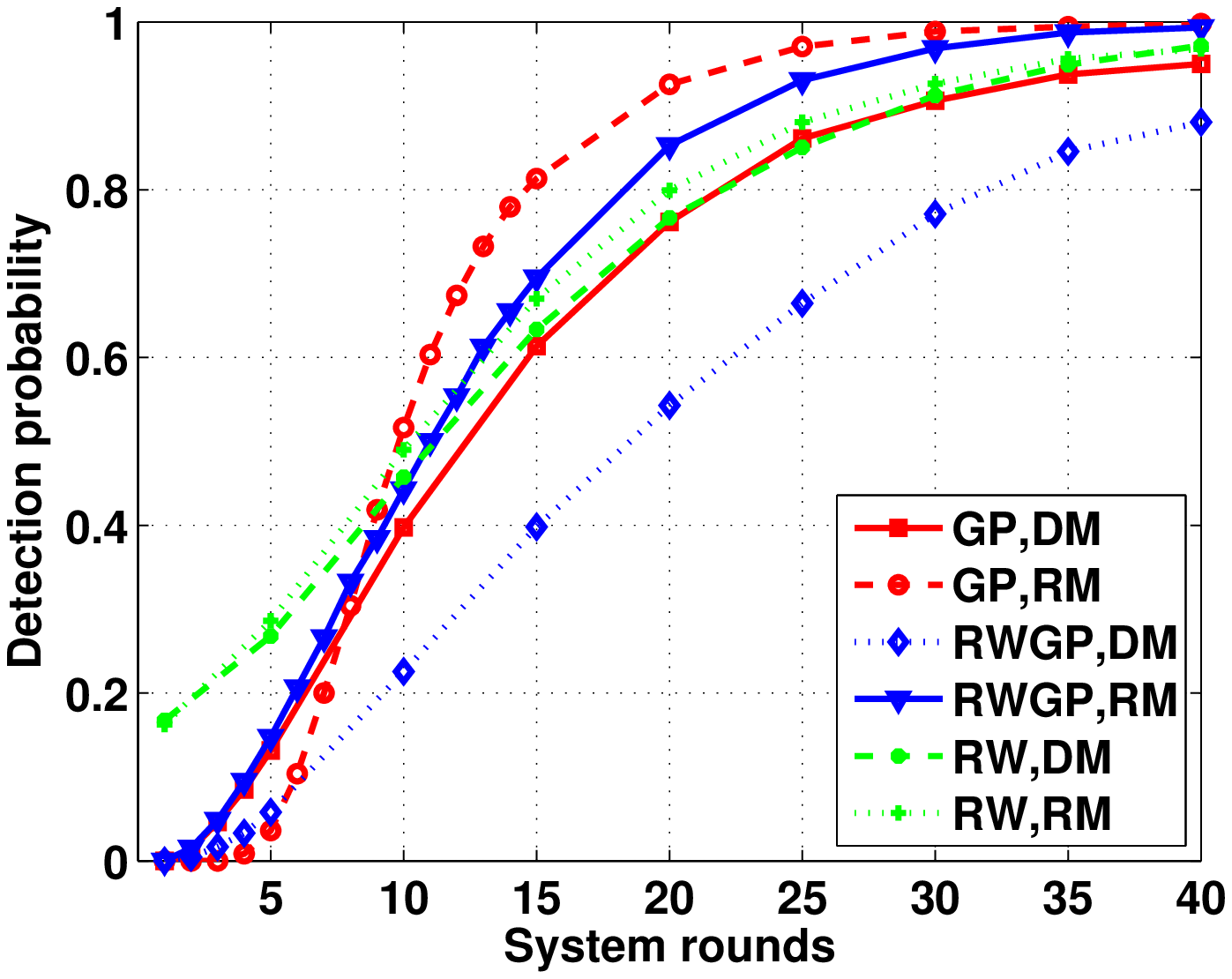}~\\
~(a)~&~(b)~\\\end{tabular}
\end{center}
\caption{Comparison in terms of detection performance for networks with $S=20$ sensors and $L=5$ master sensors. Abbrevations: GP: Proposed method, RWGP: Random Walk rounds with the gossip algorithm with pull protocol dissemination, RW: Random Walk in the network initiated at $L$ sensors. (a) fully connected sensor network; (b) irregular sensor network.}
\label{fig:20sens_dissemination}
\end{figure*}
\begin{figure*}[htb]
\begin{center}
\begin{tabular}{cc}
~\includegraphics[width=  7cm]{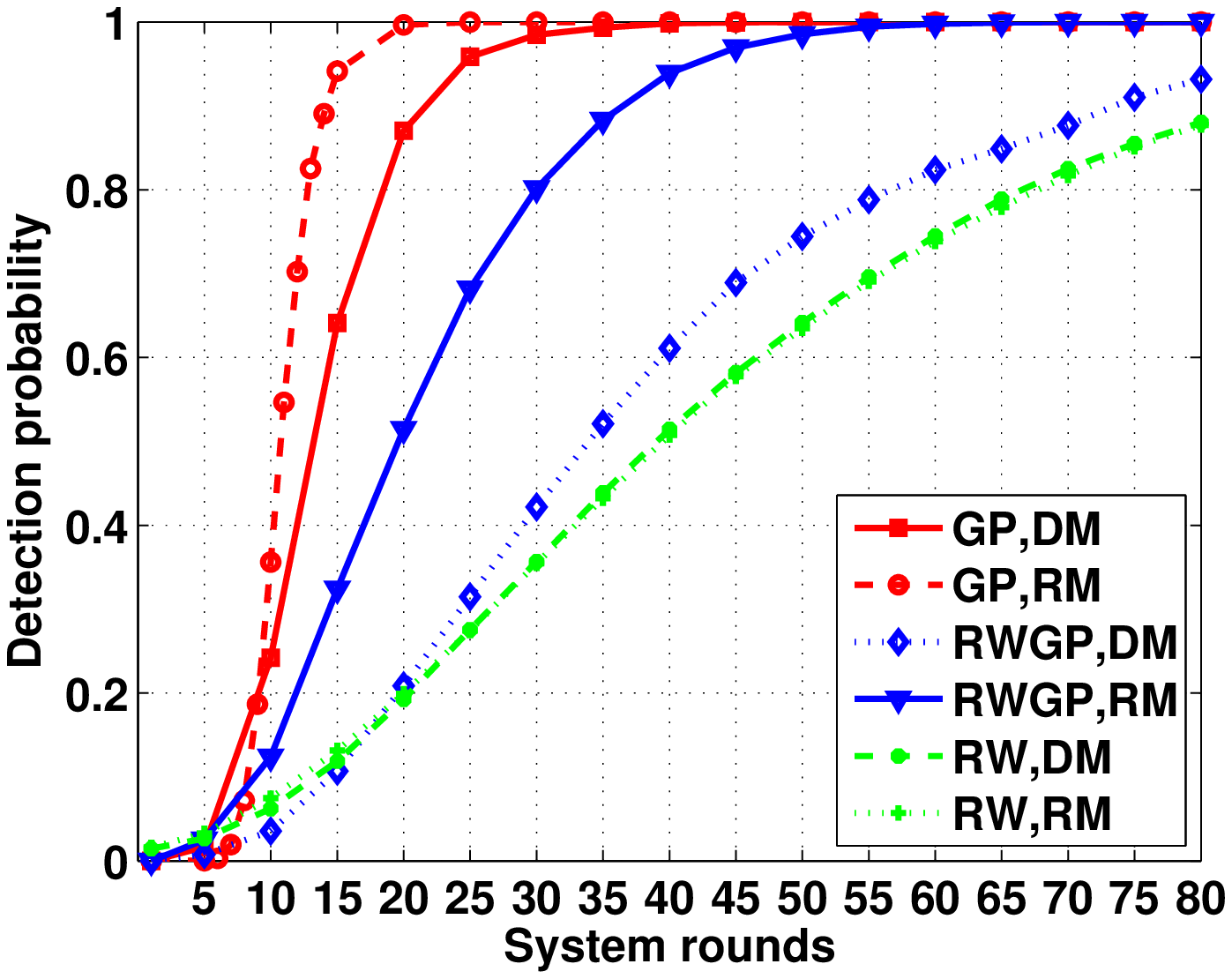}~&
~\includegraphics[width=  7cm]{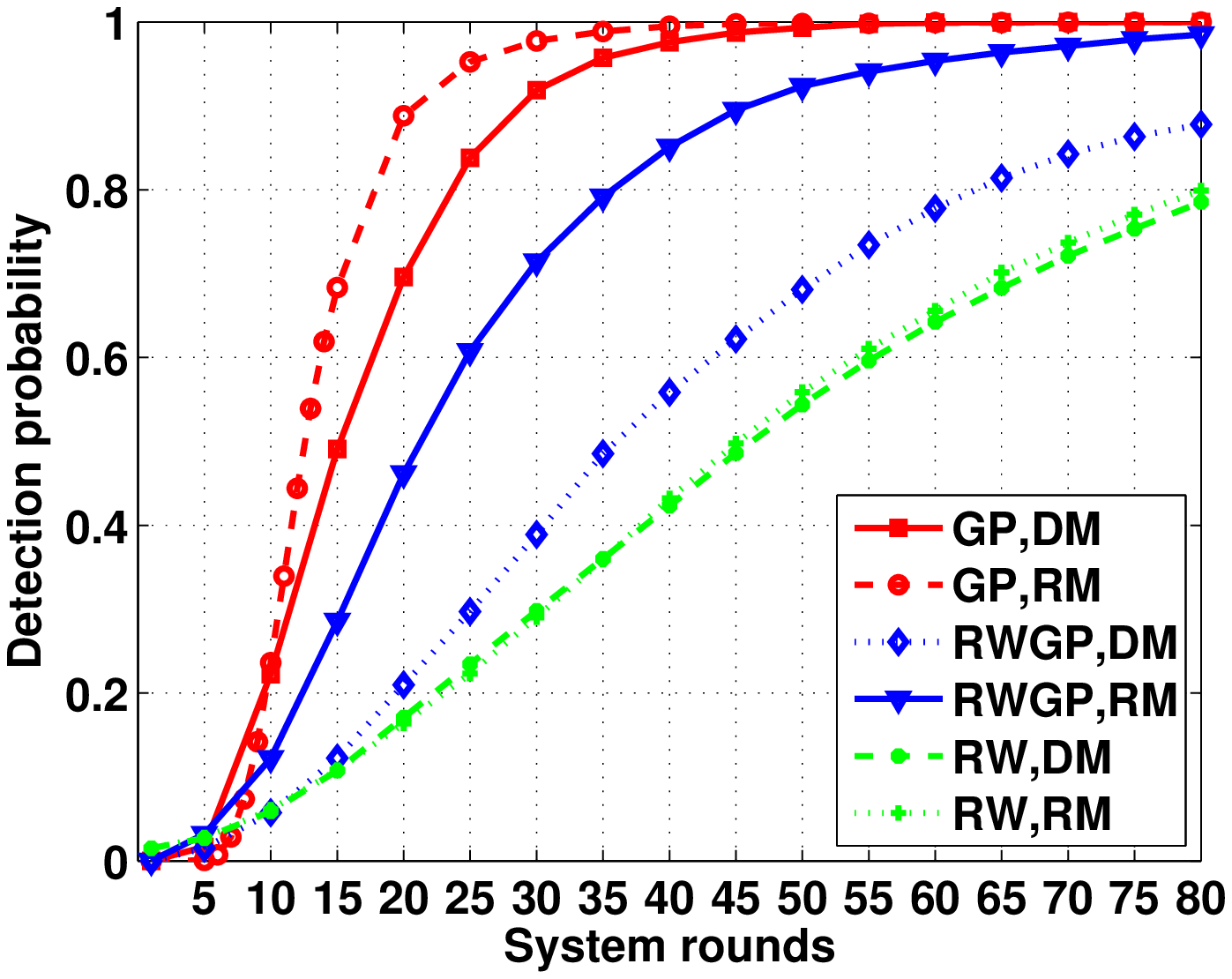}~\\
~(a)~&~(b)~\\\end{tabular}
\end{center}
\caption{Comparison in terms of detection performance for networks with $S=70$ sensors and $L=5$ master sensors. Abbrevations: GP: Proposed method, RWGP: Random Walk rounds with the gossip algorithm with pull protocol dissemination, RW: Random Walk in the network initiated at $L$ sensors. (a) fully connected sensor network; (b) irregular sensor network.}
\label{fig:70sens_dissemination}
\end{figure*}

\begin{figure*}[htb]
\begin{center}
\begin{tabular}{cc}
~\includegraphics[width=  7cm]{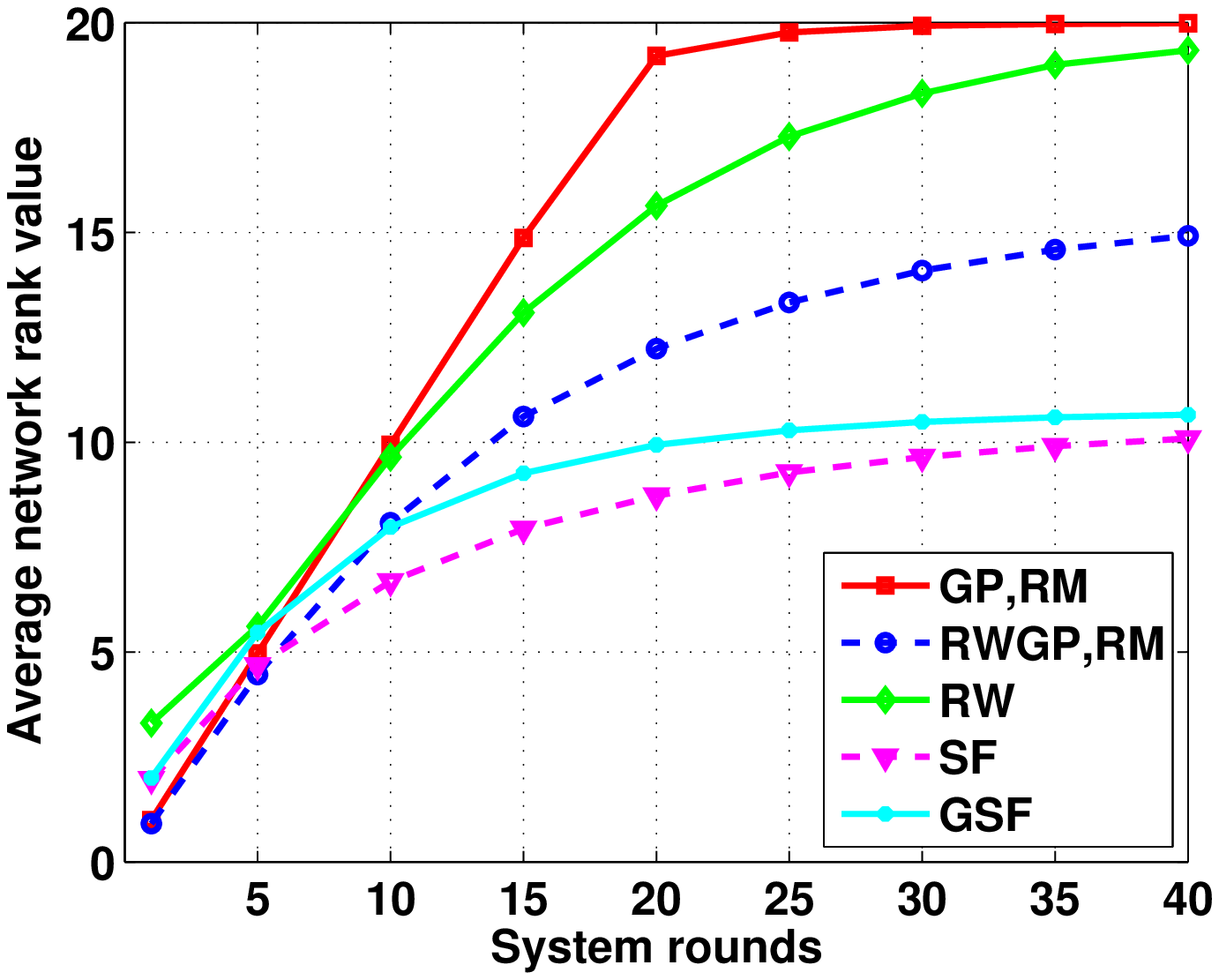}~&
~\includegraphics[width=  7cm]{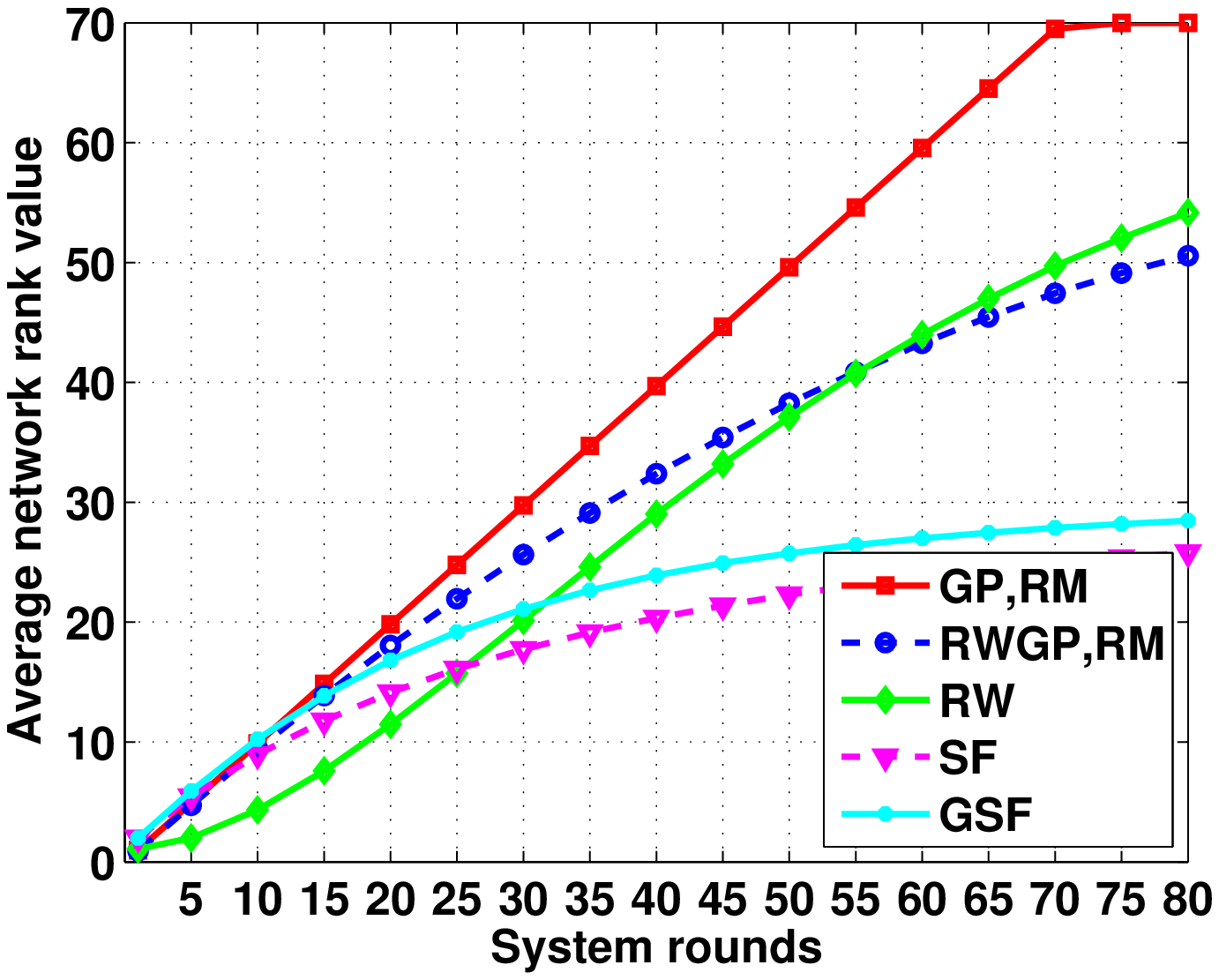}~\\
\hspace{-0.8cm}~(a)~&~(b)~\\
\end{tabular}
\end{center}
\caption{Average rank value for irregular sensor networks with $L=5$ master sensors: (a) $S=20$ sensors (b) $S=70$ sensors. Abbrevations: GP: Proposed method, RWGP: Random Walk rounds with the gossip algorithm with pull protocol dissemination, RW: Random Walk in the network initiated at $L$ sensors, SF: pull store-and-forward algorithm with a random choice of transmission message available at sensor, GSF: pull store-and-forward algorithm with a greedy choice of a transmission message available at sensor.}
\label{fig:sens_dissemination_rank}
\end{figure*}

\subsection{Communication overhead}\label{ssec:phaseI}
For the sake of completeness, we analyze the communication costs of the proposed gossiping protocol and compare it with all other schemes under comparison. Let $R_d$ and $I_d$ denote the number of bits needed for raw measurements transmission and sensor identifier, respectively. Recall that the tuple $(S, L, L_n, n, \tau)$ stands for the number of sensors in the network, the number of master sensors (clusters), the number of neighbors that each master is connected with, the average number of sensors per cluster ($n=S/L$) and the total number of transmission rounds.
 
During the first phase of GP, the master sensors receive raw measurements from their neighbors. Thus, $L_n\cdot R_d$ bits are used for communicating these values. Further, the master sensors create binary messages and send them to their neighbors. Every neighbor requires knowledge about the identifier of sensors that participate in a test, thus the cost is $I_d\cdot \lceil q(L+L_n)\rceil$ bits, plus an additional bit in each message for sending the outcome result. Hence, the overall bit consumption is $L_n R_d +L_n( I_d \lceil q(L+L_n)\rceil) +1)$. In the message exchange phase $S(1+S)$ bits are required, from which $S+1$ bits are reserved for the test outcome and the test matrix row $\mathbf{W}$. Note that this analysis includes the full vector size and it can be further compressed. The overall number of transmitted bits over $\tau$ rounds is given by:
\begin{equation}
n^b_{GP} =\tau \left[ L_n \{R_d+I_d \lceil q(L+L_n)\rceil +1\} + S(1+S)\right].
\end{equation}

We compared the communication costs of GP with the one of RWGP that takes place also in two phases. The first phase represents the random walk message collection, while the second is equivalent to the GP algorithm. Note that in the special case when RWGP and GP collect exactly the same data, they have identical decoding performance. However, if RWGP visits some sensors several times (more probable in irregular networks with a smaller connectivity degree), it performs worse than GP. In typical simulations, a random walk of RWGP terminates after $n^{th}$ transmission round, where $n$ is the number of elements per cluster in GP. RWGP transmits raw measurements, which results in $R_d+2R_d+\dots +nR_d=\frac{(1+n)R_d}{2}$ bits. Therefore, the communication cost for RWGP is given by:
\begin{equation}
n^b_{RWGP} =\tau \left[ \frac{(n+1) R_d}{2} L + S(1+S)\right].
\end{equation}

The bit transmission requirements for the $RW$ algorithm is equivalent to that of the first step of RWGP, since it transmits also raw data. The detection is performed at nodes by comparison of known sensor values at that moment, without message design step. The number of transmitted bits is equal to: $n^b_{RW} =\tau \frac{(n+1) R_d}{2} L$. Recall that for transmission of a message to all the nodes in a fully connected graph, one requires $\log{S}$ transmissions. Therefore, the SF algorithm requires in total $n^b_{SF}=\tau R_d \log S$ bits.

The comparison between the proposed method and all other  schemes regarding the bits spent for communication is illustrated in Fig. \ref{fig:comm_overhead} for a fully connected graph. Note that the proposed algorithm in this setup requires only $t=15$ rounds for efficient detection (Fig. \ref{fig:20sens_randclusterno1}), but it consumes approximately three times more communication overhead compared to that of RWGP algorithm. However, due to the specific collection approach (hops), the duration of one transmission round of RWGP lasts ten times longer than that of the proposed algorithm. From the figure we can observe that the RW algorithm has very small communication overhead. However, it requires significantly higher number of rounds ($S\log S\approx130$ rounds) compared to the detection time of the proposed GP algorithm. Overall, the proposed GP scheme is able to compete with the other schemes in terms of bits used untill detection. 

\begin{figure*}[htb]
\begin{center}
\begin{tabular}{cc}
~\includegraphics[width=  7cm]{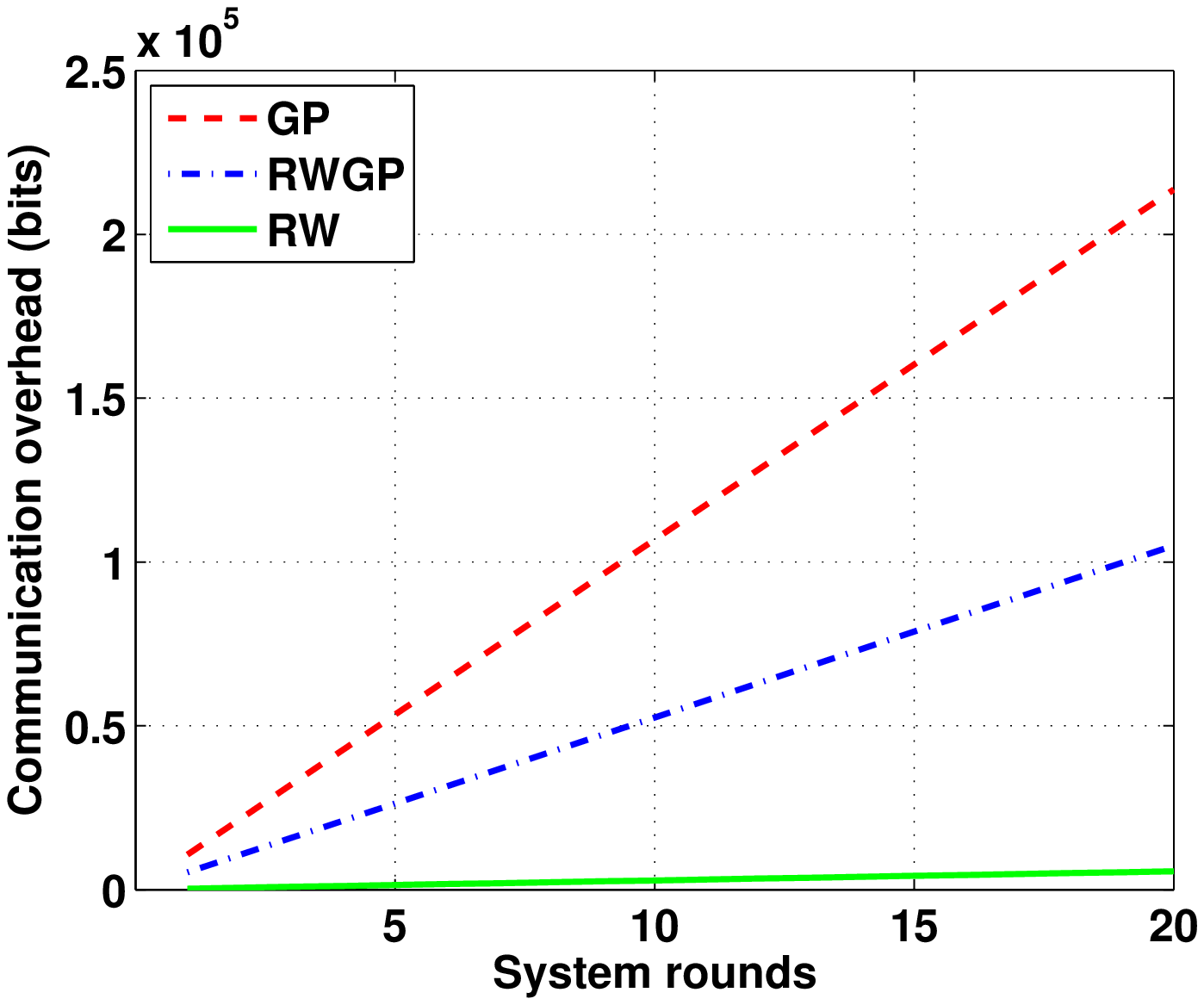}~&
~\includegraphics[width= 7cm]{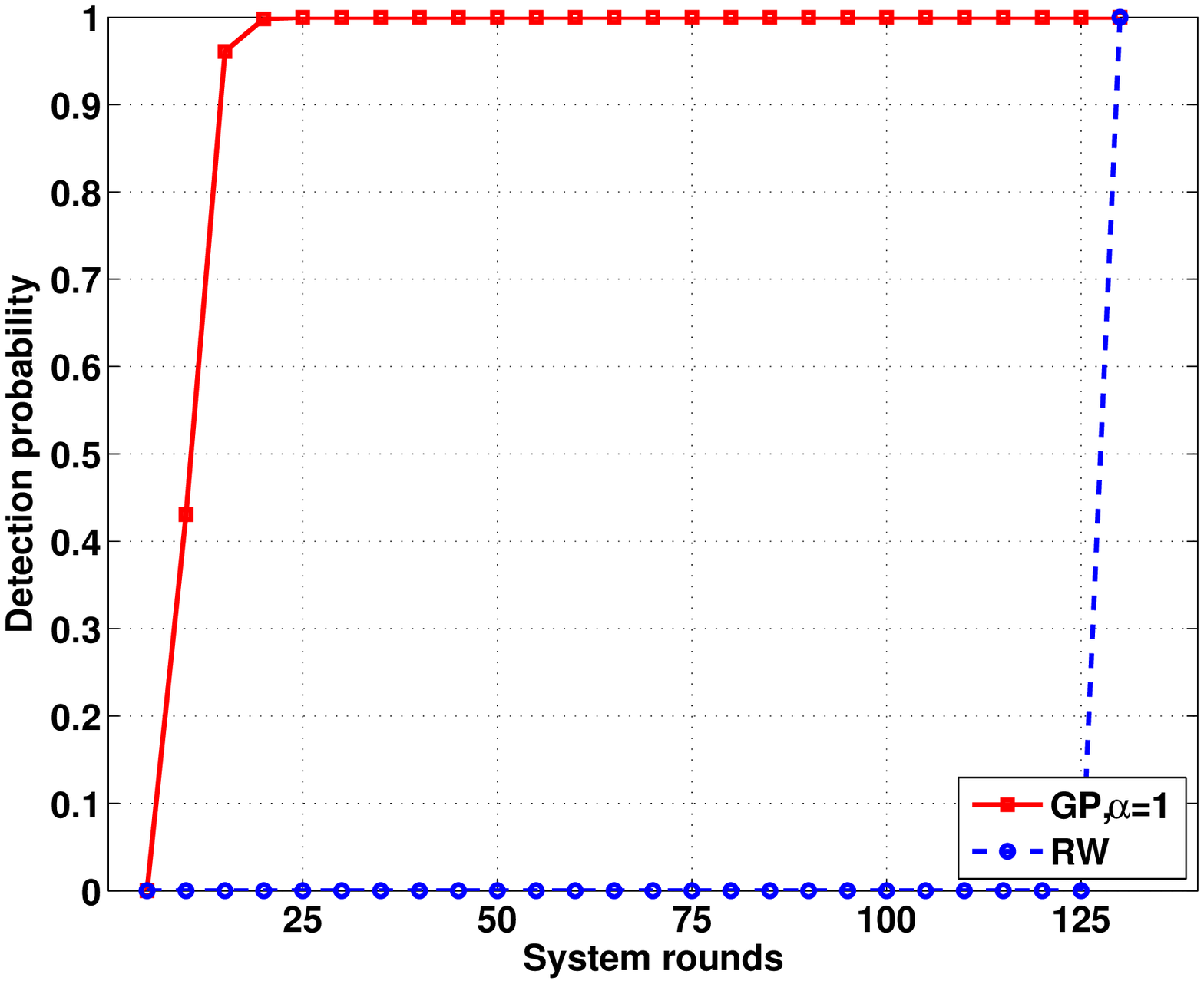}~\\
~(a)~&~(b)~\\
\end{tabular}
\end{center}
\caption{(a)Comparison of the communication overhead for several algorithms, for the following parameter values: $(S,L,L_n,\alpha,R_d,I_d,\tau)=(70,5,50,0.7,7,7,80)$. Graph is fully connected. Abbreviations: GP: Proposed method, RWGP: Random Walk rounds with gossip algorithm and pull protocol dissemination, RW: Random Walk in the network initiated at $L$ sensors. (b) Comparison of detection vs. number of rounds of the distributed detection scheme.}
\label{fig:comm_overhead}
\end{figure*}

\section{Conclusion}\label{sec:discuss}
In this work, we have addressed the problem of distributed failure detection in sensor networks. We have proposed a novel distributed algorithm that is able to detect a small number of defective sensors in a networks. We have designed a probabilistic message propagation algorithm that allows the use of a simple and efficient distance decoder at sensors. The transmitted messages are formed from local sensor observations and they are communicated using a gossip algorithm. We have derived for the worst case scenario the lower bound on the required number of linearly independent messages that sensors need to collect per cluster to ensure detection of one defective sensor with high probability. We have shown experimentally that this number is quite smaller in practice, even for the small size networks, which confirms the validity of the theoretical bound. The experimental results have shown that the proposed method outperforms other detection schemes in terms of successful detection probability. The convergence rate is very fast, which largely compensates for the higher communication overhead.

\bibliographystyle{IEEEbib}
\bibliography{refs}

\begin{thebibliography}{10}

\bibitem{Dorfman:43}
R.~Dorfman,
\newblock ``The detection of defective members of large populations,''
\newblock {\em Annals of Mathematical Statistics}, vol. 14, pp. 436--440, 1943.

\bibitem{Young:11}
M.Young and R.~Boutaba,
\newblock ``Overcoming adversaries in sensor networks: A survey of theoretical
  models and algorithmic approaches for tolerating malicious interference,''
\newblock {\em IEEE Communications Surveys and Tutorials}, vol. 13, pp.
  617--641, April 2011.

\bibitem{Chen:08}
H-B. Chen and F.~K. Hwang,
\newblock ``A survey on nonadaptive group testing algorithms through the angle
  of decoding,''
\newblock {\em J. Comb. Optim.}, vol. 15, pp. 49--59, 2008.

\bibitem{Dai:2009}
W.~Dai and O.~Milenkovic,
\newblock ``Weighted superimposed codes and constrained integer compressed
  sensing,''
\newblock {\em IEEE Trans. Inform. Theory}, vol. 55, pp. 2215--2229, May 2009.

\bibitem{DeBonis:03}
A.~De Bonis and U.~Vaccaro,
\newblock ``Constructions of generalized superimposed codes with applications
  to group testing and conflict resolution in multiple access channels,''
\newblock {\em Theor. Comput. Sci.}, vol. 306, no. 1-3, pp. 223--243, 2003.

\bibitem{Indyk:10}
P.~Indyk, H.~Q. Ngo, and A.~Rudra,
\newblock ``Efficiently decodable non-adaptive group testing,''
\newblock in {\em Proc. of the Twenty-First Annual ACM-SIAM Symposium on
  Discrete Algorithms}, 2010, pp. 1126--1142.

\bibitem{Cheraghchi:10}
M.~Cheraghchi, A.~Karbasi, S.~Mohajer, and V.~Saligrama,
\newblock ``Graph-constrained group testing,''
\newblock {\em Proc. of Int. Symp. on Inform. Theory (ISIT)}, pp. 1913--1917,
  2010.

\bibitem{Mezard:07}
M.~M\'{e}zard and C.~Toninelli,
\newblock ``Group testing with random pools: Optimal two-stage algorithms,''
\newblock {\em IEEE Trans. Inform. Theory}, vol. 57, no. 3, pp. 1736--1745,
  March 2011.

\bibitem{Hong:04}
Y.-W. Hong and A.~Scaglione,
\newblock ``Group testing for sensor networks: The value of asking the right
  question,''
\newblock {\em 38th Asilomar Conference on Signals, Systems and Computers},
  vol. 2, pp. 1297--1301, 2004.

\bibitem{Varshney:97}
P.~K. Varshney,
\newblock {\em Distributed Detection and Data Fusion},
\newblock Springer-Verlag New York, Inc., 1st edition, 1996.

\bibitem{Tsitsiklis:93}
J.~N. Tsitsiklis,
\newblock ``Decentralized detection,''
\newblock {\em Proc. of Advanced Statistical Signal Processing}, vol. 2-Signal
  Detection, pp. 297--344, 1993.

\bibitem{Tian:07}
Q.~Tian and E.~J. Coyle,
\newblock ``Optimal distributed detection in clustered wireless sensor
  networks,''
\newblock {\em IEEE Trans. on Signal Proc.}, vol. 55, no. 7, pp. 3892--3904,
  2007.

\bibitem{Viswanathan:97}
R.~Viswanathan and P.~K. Varshney,
\newblock ``Distributed detection with multiple sensors: Part
  {I}-{F}undamentals,''
\newblock {\em Proc. IEEE}, vol. 85, no. 1, pp. 54--63, Jan. 1997.

\bibitem{Blum:97}
R.~S. Blum, S.~A. Kassam, and H.~V. Poor,
\newblock ``Distributed detection with multiple sensors: Part {II}-{A}dvanced
  topics,''
\newblock {\em Proc. IEEE}, vol. 85, no. 1, pp. 64--79, Jan. 1997.

\bibitem{Dimakis:10}
A.~Dimakis, S.~Kar, J.M.F. Moura, M.G. Rabbat, and A.~Scaglione,
\newblock ``Gossip algorithms for distributed signal processing,''
\newblock {\em Proc. IEEE Trans. Inform. Theory}, vol. 98, pp. 1847--1864, Nov.
  2010.

\bibitem{Demers:87}
A.~Demers, D.~Greene, C.~Hauser, W.~Irish, J.~Larson, S.~Shenker, H.~Sturgis,
  D.~Swinehart, and D.~Terry,
\newblock ``Epidemic algorithms for replicated database maintenance,''
\newblock pp. 1--12, 1987.

\bibitem{Karp00}
R.~Karp, C.~Schindelhauer, S.~Shenker, and B.~V\"{o}cking,
\newblock ``Randomized rumor spreading,''
\newblock pp. 565--574, 2000.

\bibitem{Deb:2006}
S.~Deb, M.~Medard, and C.~Choute,
\newblock ``Algebraic gossip: A network coding approach to optimal multiple
  rumor mongering,''
\newblock {\em IEEE Trans. Inform. Theory}, vol. 52, no. 6, pp. 2486--2507,
  2006.

\bibitem{Cheraghchi:11}
M.~Cheraghchi, A.~Hormati, A.~Karbasi, and M.~Vetterli,
\newblock ``Group testing with probabilistic tests: Theory, design and
  application,''
\newblock {\em IEEE Trans. of Inf. Theory}, vol. 57, no. 10, pp. 7057--7067,
  Oct. 2011.

\bibitem{Gallager:62}
R.~Gallager,
\newblock ``Low-density parity-check codes,''
\newblock {\em Monograph, M.I.T. Press}, 1963.

\end{thebibliography}

%
\appendix \section{Appendix}\label{sec:appendix}
\subsection{Model for probability $P(q|m)$}\label{ssec:p(q|m)}

$P(q|m)$ models the probability of event that multiple defective sensors are present in the same cluster but only a subset of defective sensors participates in the test. This event introduces errors while detection of defective sensors. Recall that sensors participate in the test with the probability $q$. For $m$ defective sensors possible message realizations are given with elements of the polynomial $(q+(1-q))^m$. This polynomial represents the the binomial expansion of the form $(x+y)^m$, with $x=q$ and $y=(1-q)$. Polynomial expansion is equal to $(x+y)^m=c_0 x^m+c_1 xy^{m-1}+\dots+c_m y^m$ and the coefficients $c_i=\binom{m}{i}$ represent the numbers of $i$-th row of Pascal's triangle. Messages that do not cause decoding error are the messages of all zeros and of all ones. These messages occur with probabilities $q^m$ and $(1-q)^m$, respectively and they have coefficients equal to $1$. Note that $(q+(1-q))^m=1^m=1$ and that probability of error event is therefore equal to:
\begin{equation}
P(q|m) =\frac{1-q^m-(1-q)^m}{(q+(1-q))^m} = 1-q^m-(1-q)^m.
\end{equation}

\end{document}